\documentclass[a4paper,USenglish,numberwithinsect]{lipics}
\pdfoutput=1

\usepackage{microtype}

\usepackage[draft]{fixme} 

\usepackage{tabularx}

\usepackage{booktabs}
\usepackage{multirow}

\usepackage[algoruled]{algorithm2e} 

\usepackage{amsmath,amssymb,amsbsy}

\DeclareMathOperator{\diam}{diam}
\DeclareMathOperator{\diag}{diag}
\DeclareMathOperator{\degree}{deg}
\DeclareMathOperator{\adj}{Adj}
\DeclareMathOperator{\EV}{\mathbf{E}}
\DeclareMathOperator{\VAR}{\mathbf{Var}}
\DeclareMathOperator{\PR}{\mathbf{Pr}}
\newcommand{\smax}{s_\mathrm{max}}
\newcommand{\smin}{s_\mathrm{min}}

\renewcommand{\vec}{\mathbf}

\newcommand{\SSum}{\mathcal{S}}
\newcommand*\bell{\ensuremath{\boldsymbol\ell}}
\newcommand{\genS}{\! S}

\newcommand{\poly}{\operatorname{poly}}

\usepackage{xspace}
\newcommand{\eapx}{\mbox{$\varepsilon$-approximate}\xspace}

\theoremstyle{remark}
\newtheorem{observation}[theorem]{Observation}

\begin{document}
\title{Distributed Selfish Load Balancing with Weights and Speeds}
\author[1,*]{Clemens Adolphs}
\author[2]{Petra Berenbrink}
\affil[1]{RWTH Aachen University, Aachen, Germany}
\affil[2]{Simon Fraser University, Burnaby, Canada}


\maketitle

\let\oldthefootnote\thefootnote
\renewcommand{\thefootnote}{\fnsymbol{footnote}}
\footnotetext[1]{The research was carried out during a visit to SFU}
\let\thefootnote\oldthefootnote

\subparagraph{Subject:} Distributed Algorithms

\begin{abstract}
  In this paper we consider neighborhood load balancing in the
  context of selfish clients.  We assume that a network of $n$
  processors and $m$ tasks is given.  The processors may have
  different speeds and the tasks may have different weights.  Every
  task is controlled by a selfish user. The objective of the user is
  to allocate his/her task to a processor with minimum load.

  We revisit the concurrent probabilistic protocol introduced in
  \cite{Berenbrink2011}, which works in sequential rounds.  In each
  round every task is allowed to query the load of one randomly chosen
  neighboring processor. If that load is smaller the task will
  migrate to that processor with a suitably chosen probability.  Using
  techniques from spectral graph theory we obtain upper bounds on the
  expected convergence time towards approximate and exact Nash
  equilibria that are significantly better than the previous results
  in \cite{Berenbrink2011}.  We show results for uniform tasks on
  non-uniform processors and the general case where the tasks have
  different weights and the machines have speeds.  To the best of our
  knowledge, these are the first results for this general setting.
\end{abstract}

\section{Introduction}
Load Balancing is an important aspect of massively parallel
computations as it must be ensured that resources are used to their
full efficiency.  Quite often the major constraint on balancing
schemes for large networks is the requirement of \emph{locality} in
the sense that processors have to decide if and how to balance their
load with local load information only.  Global information is often
unavailable and global coordination usually very expensive and
impractical. Protocols for load balancing should respect this locality
and still guarantee fast convergence to balanced states where every
processor has more or less the same load.

In this paper we consider neighborhood load balancing in a selfish
setting.  We assume that a network of $n$ processors and $m$ tasks is
given.  The processors can have different speeds and the tasks can
have different weights.  Initially, each processor stores some number
of tasks. The total number of tokens is time-invariant, i.e., neither
do new tokens appear, nor do existing ones disappear. The {\em load}
of a node at time $t$ is the total weight of all tasks assigned to
that node at that time.

Every task is assumed to belong to a selfish user. The goal of the
user is to allocate the task to a processor with minimum load.  We
assume neighborhood load balancing, meaning that task movements are
restricted by the network.  Users that are assigned to the processor
represented by node $v$ of the network are only allowed to migrate
their tasks over to processors that are represented by neighboring
nodes of $v$. Hence, the network models load balancing restrictions.
Our model can be regarded as the selfish version of diffusion load
balancing.

In this paper we revisit the concurrent probabilistic protocol
introduced in \cite{Berenbrink2011}. The load balancing process works
in sequential rounds.  In each round every task is allowed to check
the load of one randomly chosen neighboring processor. If that load
is smaller the task will migrate to that processor with a suitably
chosen probability. Note that, if the probability is too large (for
example all tasks move to a neighbor with smaller load) the system
would never be able to reach a balanced state. Here, we chose the
migration probability as a function of the load difference of the two
processors. No global information is necessary.

Using techniques from spectral graph theory similar to those used in
\cite{Elsasser2002}, we can calculate upper bounds on the expected
convergence time towards approximate and exact Nash equilibria that are
significantly better than the previous results in
\cite{Berenbrink2011}.  We show results for uniform tasks on
non-uniform processors and the general case where the tasks have
different weights and the machines have speeds.  To our best knowledge
these are the first results for this general setting.  For weighted
tasks we deviate from the protocol for weighted tasks given in
\cite{Berenbrink2011}. In our protocol, a player will move from one
node to another only if the player with the largest weight would also
do so.  It is also straightforward to apply our techniques to discrete
diffusive load balancing where each node sends the rounded expected
flow of the randomized protocol to its neighbors (\cite{AdolphsIPDPS12}).

\subsection{Model and New Results}
The computing network is represented by an undirected graph $G = (V,E)$ with
vertices representing the processors and edges representing the direct
communication links between them. The number of processors  $n = |V|$ and 
the number of tasks is $m$.
The \emph{degree} of a vertex $v \in V$ is $\deg(v)$. 
The maximum degree of the network is denoted by
$\Delta$, and for two nodes $v$ and $w$ the maximum of $\deg(v)$ and $\deg(w)$
is $d_{vw}$.

$s_i \in \mathbb{R}$ is the speed of processor $i$. We assume that the
speeds are scaled so that the smallest speed, called $\smin$, is
$1$. If all speeds are the same we say the speeds are {uniform}.  Let
$\SSum = \sum_{i\in V} s_i$ If all $s_i$ be the total capacity of the
processors.  Define $\smax$ as the maximum speed and $\smin$ as the
minimum speed of the processors.  In the case of weighted task task
$\ell$ has a weight $w_l \in (0,1]$.  In the case of uniform tasks we
assume the weight of all tasks is one.  Let $W$ denote the total sum
of all weights, $W = \sum_{i} W_i(x)$.

A \emph{state} $x$ of the system is defined by the distribution of
tasks among the processors.  For the case of uniform indivisible
tasks, we denote with $w_i(x)$ the number of tasks on processor $i$ in
state $x$. For the case of weighted tasks, $W_i(x)$ denotes the total
weight on processor $i$ whereas $w_l \in (0,1]$ denotes the weight of
tasks $l$.  The {\em load} of a processor $i$ is defined as
$w_i(x)/s_i$ in the case of uniform tasks and as $W_i(x)/s_i$ in the
case of weighted tasks.  The goal is to reach a state $x$ in which no
task can benefit from migrating to a neighboring processor. Such s
state is called {\em Nash Equilibrium}.

\subsubsection{Uniform Tasks  on Machines with Speeds}
For uniform tasks, one round of the protocol goes as follows.  Every
task selects a neighboring node uniformly at random.  If migrating to
that node would lower the load experienced by the task, the task
migrates to that node with proportional to the load difference and the
speeds of the processors.  For a detailed description of the protocol
see Algorithm~\ref{alg:DSL} in Section \ref{sec:selfish}.

The first result concerns convergence to an approximate Nash
equilibrium if the number of tasks, $m$, is large enough. For a
detailed definition of Laplacian matrix see Section \ref{sec:Models}.
\begin{theorem}
    \label{thm:ApproxNE}
    Let $\psi_c = 16n \cdot \Delta \cdot \smax / \lambda_2$ and let
    $\lambda_2$ denote the second smallest eigenvalue of the network's
    Laplacian matrix. Then, Algorithm~\ref{alg:DSL}
    (p.~\pageref{alg:DSL}) reaches a state $x$ with $\Psi_0(x) \leq
    4\cdot\psi_c$ in expected time
    \begin{equation*}
        \mathcal{O}\left( \ln \left(\frac{m}{n}\right) \cdot
          \frac{\Delta}{\lambda_2}
          \cdot \smax^2\right).
    \end{equation*}

    If $m \ge 8\cdot\delta\cdot \smax \cdot \SSum\cdot n^2$
    for some $\delta > 1$, this state is an $\eapx$-Nash equilibrium with 
    $\varepsilon = 2/(1+\delta)$.
\end{theorem}

From the state reached in Theorem \ref{thm:ApproxNE}, we then go on to prove the
following bound for convergence to a Nash equilibrium.
\begin{theorem}
    \label{thm:TNE}
    Let $\psi_c$ be defined as in Theorem~\ref{thm:ApproxNE}, and
    let $T$ be the first time step in which the system
    is in a Nash equilibrium.
    Under the condition that the speeds $s_i$ are integer multiples of a common 
    factor, $\epsilon$, it holds
    \begin{equation*}
        \EV[T] = \mathcal{O}\left( 
          n \cdot \frac{\Delta^2}{\lambda_2} \cdot 
          \frac{\smax^4}{\epsilon^2} \right).
    \end{equation*}
\end{theorem}
These theorems are proven in Section \ref{sec:selfish}.
Our bound of  Theorem \ref{thm:TNE} is smaller by at least a factor of
$\Omega (\Delta \cdot \diam(G))$ than
the bound found in \cite{Berenbrink2011} (see Observation \ref{obs1}). 

\begin{table}[hbtp]  
\caption{Comparison with existing results \label{tab:Examples}}
\begin{center}
  \begin{tabular}{ccccc}\toprule%
    \multirow{2}*{Graph} 
    & \multicolumn{2}{c}{\eapx NE} & \multicolumn{2}{c}{Nash Equilibrium}\\
    \cmidrule(l){2-3}\cmidrule(l){4-5}
    & This Paper & \cite{Berenbrink2011} & 
    This Paper & \cite{Berenbrink2011}\\\midrule
    Complete Graph & 
    $\ln\left(\frac{m}{n}\right)$ & 
    $n^2 \cdot \ln(m)$ & 
    $n^2$ &
    $n^6$
    \\\midrule%
    Ring, Path     & 
    $n^2 \cdot \ln\left(\frac{m}{n}\right)$ & 
    $n^3 \cdot \ln(m)$ & 
    $n^3$ & 
    $n^5$
    \\\midrule%
    Mesh, Torus    & 
    $n \cdot \ln\left(\frac{m}{n}\right)$ & 
    $n^2 \cdot \ln(m)$ & 
    $n^2$ & 
    $n^4$
    \\\midrule%
    Hypercube      & 
    $\ln(n)\cdot \ln\left(\frac{m}{n}\right)$ & 
    $n\cdot \ln^3(n) \cdot \ln(m)$ &
    $n\cdot \ln^2(n)$ & 
    $n^3 \cdot \ln^5(n)$
    \\\bottomrule%
  \end{tabular}
\end{center}
\end{table}

We summarize the results for the most important graph classes in Table
\ref{tab:Examples}. The table gives an overview of asymptotic bounds
on the expected runtime to reach an approximate or a exact Nash
equilibrium. We omit the speeds from this table because they are
independent of the graph structure and, therefore, the same for each
column. We compare the results of this paper to the bounds obtained
from 
\cite{Berenbrink2011}. These contain a factor $\SSum = \sum_i s_i$,
which we replace with $n$, using $\SSum = \sum_i s_i \ge n$.  The
table shows that for the graph classes at hand, our new bounds are
superior to those in \cite{Berenbrink2011}.

\subsubsection{Weighted Tasks on Machines with Speeds}
In Section \ref{sec:weighted}, we study a slightly modified protocol
(see \ref{alg:DSL_weighted}) that allows tasks only to migrate to a
neighboring processor if that would decrease their experienced load
by a threshold depending on the speed of the processors. This protocol
allows the tasks only to reach an approximate Nash Equilibrium.

\begin{theorem}\label{thm:ANE_weight}
  Let 
$\psi_c = 16\cdot n\cdot \Delta / \lambda_2 \cdot \smax / \smin^2$
and let $\lambda_2$ denote the second smallest eigenvalue of the network's
Laplacian matrix. Then, Algorithm
\ref{alg:DSL_weighted} (p.~\pageref{alg:DSL_weighted}) reaches a state $x$
with $\Psi_0(x) \leq 4\cdot \psi_c$ in time 
\begin{equation*}
  \mathcal{O}\left(
    \ln \left(\frac{m}{n}\right) \cdot \frac{\Delta}{\lambda_2} \cdot 
    \frac{\smax^2}{\smin}
    \right).
\end{equation*}

Under the condition that $W > 8\cdot \delta \cdot \smax/\smin \cdot \mathcal{S}
\cdot n^2$ for some 
$\delta > 1$, this state is an $2/(1+\delta)$-approximate Nash equilibrium.
\end{theorem}
For the case of uniform speeds the theorem gives a bound of   
$\mathcal{O}\left(\ln (m/n) \cdot \Delta/\lambda_2\right) $
for the convergence time.

\paragraph{Outline.} After presenting the notation and
preliminaries in Section \ref{sec:Models}, we treat the case of machines with
speeds in Section \ref{sec:selfish}. 
Section \ref{sec:weighted} treats the case
of weighted tasks. Proofs are found in the appendix.

\subsection{Related Work}
The work closest to ours is in \cite{Berenbrink2007a,
  BerenbrinkESA07,Berenbrink2011}.  \cite{Berenbrink2007a} considers the
case of identical machines in a complete graph.  The authors introduce
a protocol similar to ours that reaches a Nash Equilibria (NE) in time
${\mathcal O}( \log \log m + \poly(n))$. Note that for complete graphs
the NE and the optima (where the load discrepancy is zero or one) are
identical. An extension of this model to weighted tasks is studied
in~\cite{BerenbrinkESA07}. Their protocol converges to a NE in time
polynomially in $n$, $m$, and the largest weight.  In
\cite{Berenbrink2011} the authors consider a model similar to ours,
meaning general graphs with processors with speed and weighted tasks.
They use a potential function similar to ours for the analysis.  The
potential drop is linked to the maximum load deviation, $L_\Delta =
\max_{i \in V} |e_i/s_i|$. The authors show that an edge must exist
over which the load difference is at least $L_\Delta / \diam(G)$. As
long as the potential $\Psi_0$ is large enough, it can then be shown
that there is a multiplicative drop.  This is then used to prove
convergence to an approximate Nash equilibrium. Subsequently, a
constant drop in $\Phi_1$ is used to finally converge to a Nash
equilibrium. The two main results of \cite{Berenbrink2011} for
machines with speeds are presented in Table \ref{tab:Examples}

Our paper relates to a general stream of works for selfish load
balancing on a complete graph. There is a variety of issues that have
been considered, starting with seminal papers on algorithms and
dynamics to reach NE~\cite{EvenDarTALG07,Feldmann03}. More directly
related are concurrent protocols for selfish load balancing in
different contexts that allow convergence results similar to
ours. Whereas some papers consider protocols that use some form of
global information~\cite{Even-Dar2005} or coordinated
migration~\cite{FotakisATOM10}, others consider infinitesimal or
splittable tasks~\cite{FischerTCS08,AwerbuchSODA08} or work without
rationality assumptions~\cite{Fischer08,AckermannSPAA09}. The machine
models in these cases range from identical, uniformly related (linear
with speeds) to unrelated machines. The latter also contains the case
when there are access restrictions of certain agents to certain
machines.  For an overview of work on selfish load balancing
see, e.g.,~\cite{Vocking2007own}.

Our protocol is also related to a vast amount of literature on
(non-selfish) load balancing over networks, where results usually
concern the case of identical machines and unweighted tasks.  In
expectation, our protocols mimic continuous diffusion, which has been
studied initially in~\cite{Cybenko1989,Boillat1990} and later, e.g.,
in~\cite{Muthukrishnan1998}. This work established the connection between
convergence, discrepancy, and eigenvalues of graph matrices. Closer to
our paper are discrete diffusion processes -- prominently studied
in~\cite{Rabani1998}, where the authors introduce a general technique to
bound the load deviations between an idealized and the actual
processes. Recently, randomized extensions of the algorithm
in \cite{Rabani1998} have been considered,
e.g.,~\cite{Elsaesser06,Friedrich2009}.
\section{Notation and Preliminaries}
\label{sec:Models}
In this section we will give the more technical definitions.

A \emph{state} $x$ of the system
is defined by the distribution of tasks among the processors.  For the case of
uniform indivisible tasks, we denote with $w_i(x)$ the number of tasks on
processor $i$ in state $x$. For the case of weighted tasks, $W_i(x)$ denotes the
total weight on processor $i$ whereas $w_l \in (0,1]$ denotes the weight of
tasks $l$.  

The \emph{task vector} is defined as $ \vec{w}(x) = (w_1(x), w_2(x),
\cdots w_n(x))^\top.  $ We define the \emph{load} of processor $i$ in state $x$
as $ \ell_i(x) := w_i(x)/s_i.  $ In analogy to the task vector, we define the
\emph{load vector} as $ \vec{\bell}(x) = (\ell_1(x), \cdots \ell_n(x))^\top.  $
For the processor speeds, we define the \emph{speed vector} as $ \vec{s} = (s_1,
\cdots s_n)^\top $ and the \emph{speed matrix} as $ S \in \mathbb{N}^{n\times
  n}, \quad S_{ii} = s_i.  $ Let $\smax := \max_{i\in V} s_i$ denote the maximum
speed.  The task
vector $\vec{w}(x)$ and the load vector $\vec{\bell}(x)$ are related by the
speed matrix $S$ via $\vec{\bell}(x) = S^{-1} \vec{w}(x).$ 
The average load of the network is $\bar \ell_i = m/\SSum$. In the completely
balanced state, each node has exactly this load. The corresponding task vector
is $\vec{\bar w} = m/\SSum\cdot \vec{s}$ and we define $\vec{e}(x)$ of the
deviation of the actual task vector from the average load vector,
$\vec{e}(x) = \vec{w}(x) - \vec{\bar w}$. It is clear that 
$\sum_{i\in V} e_i = 0$.

A state $x$ of the system is called a \emph{Nash equilibrium} (NE) 
if no single task
can improve its perceived load by migrating to a neighboring node while all
other tasks remain where they are, i.e., $\ell_i - \ell_j \leq 1/s_j$ for all edges $(i,j)$.
A state $x$ of the system is called an $\varepsilon$-\emph{approximate Nash
  equilibrium} (\eapx-NE) if no single task can improve its perceived load by a
factor of $(1-\varepsilon)$, i.e. 
$(1-\varepsilon) \cdot \ell_i - \ell_j \leq 1/s_j.$


The Laplacian $L(G)$ is a matrix widely used in graph theory. It is the $n
\times n$ matrix whose diagonal elements are $L_{ii} = \deg(i)$, and the
off-diagonal elements are $L_{ij} = -1$ if $(i,j) \in E(G)$ and $0$
otherwise. The generalized Laplacian $LS^{-1}$, where 
$S$ is the diagonal matrix containing the speeds $s_i$ \cite{Elsasser2002},
is used to analyze the behavior of migration in heterogeneous networks.


\section{Uniform Tasks on Machines with Speeds}
\label{sec:selfish}
The pseudo-code of our protocol is given in Algorithm~\ref{alg:DSL}.
Recall that $d_{i,j}$ is defined as $\max \{ \degree(i), \degree(j) \}$.
$\alpha$ is defined as $4\smax$.\label{def:alpha} 

\begin{algorithm}[H]
\DontPrintSemicolon
\Begin{%
\ForEach{task $\ell$ in parallel}{%
Let $i = i(l)$ be the current machine of task $l$\;
Choose a neighboring machine $j$ uniformly at random\;
\If{$\ell_i - \ell_j > 1/s_j$}{%
  Move task $\ell$ from node $i$ to node $j$ with probability
  \[ p_{ij}:= \frac{\degree(i)}{d_{i,j}} \cdot
  \frac{\ell_i - \ell_j}{\alpha \cdot \left( \frac{1}{s_i} +
  \frac{1}{s_j} \right)\cdot W_i} \]
}
}
}
\caption{Distributed Selfish Load Balancing\label{alg:DSL}}
\end{algorithm}

The analysis of this protocol initially 
follows the steps of \cite{Berenbrink2011} up to
Lemma 3.3 (Restated as Lemma \ref{thm:Lemma35} in the appendix). Before we
outline the remainder of our proof, we introduce some more notation.
\begin{definition}
  \label{def:fij}
  For a given state $x$, we define $f_{ij}(x)$ as the \emph{expected flow} over
  edge $(i,j)$. It holds
  \begin{equation*}
    f_{ij}(x) = \begin{cases} \displaystyle
      \frac{\ell_i(x) - \ell_j(x)}{\alpha\cdot d_{ij}\cdot \left(
          \frac{1}{s_i} + \frac{1}{s_j}
        \right)} & \text{if}\quad \ell_i(x) - \ell_j(x) > \frac{1}{s_j} \\
      0 & \text{otherwise}.\end{cases}
    \label{eqn:fij}
 \end{equation*}
\end{definition}
The following two potential functions will be used  in the analysis.
\begin{definition}
For $r = 0,1$, define 
\begin{equation*}
  \Phi_r(x) := \sum_{i \in V} \frac{W_i(x) \cdot (W_i(x) + r)}{s_i}.
\end{equation*}
\end{definition}
The potential $\Phi_0$ is minimized for the average task vector, 
$\vec{\bar w}$. We define the according normalized potential $\Psi_0$.
\begin{definition} The normalized potential $\Psi_0(x)$ is defined as 
  \begin{equation*}
    \Psi_0(x) = \Phi_0(x) - \frac{m^2}{\SSum} = 
    \sum_{i \in V} \frac{e_i(x)^2}{s_i}.
  \end{equation*}
\end{definition}
We want to relate this potential function to the load imbalance in the system.
To this end, we define a new quantity.
\begin{definition}
  \label{def:LDelta}
    We define the \emph{maximum load difference} as
    \begin{equation*}
        L_\Delta(x) = \max_{i \in V} \left| \frac{W_i(x)}{s_i}
        - \frac{m}{\SSum} \right | = \max_{i\in V} \left| \frac{e_i}{s_i} \right|.
   \end{equation*}
\end{definition}
\begin{definition}
  \label{def:DeltaPhi}
  Let $t > 0$ be some time step during the executing of our protocol and let
  $X^t$ denote the state of the system at that time step.  We define
  $\Delta \Phi_r(X^t) := \Phi_r(X^{t-1}) - \Phi_r(X^t)$
  as the drop in potential $\Phi_r(X^t)$ in time step $t$.  The sign convention
  for $\Delta \Phi_r(X^t)$ is such that a \emph{drop} in $\Phi_r(x)$ from time
  step $t-1$ to $t$ gets a \emph{positive} sign. This emphasizes that a large
  drop in $\Phi_r(x)$ is a desirable outcome of our process.
  $\Delta \Psi_0(X^t)$ is defined analogously.
\end{definition}
\begin{lemma}
  \label{thm:Psi0properties}
  The shifted potential $\Psi_0(x)$ has the following properties.
  \begin{enumerate}
  \item[(1)] The change in $\Psi_0(x)$ due to migrating tasks is the same as the
    change in $\Phi_0(x)$, i.e. 
    \begin{equation*}
      \Delta \Psi_0(X^t|X^{t-1}=x) = 
      \Delta \Phi_0(X^t|X^{t-1}=x)
    \end{equation*}
  \item[(2)] The potential $\Psi_0(x)$ can also be written using the generalized
    dot-product introduced in Section \ref{sec:genlap}, 
    $\Psi_0(x) = \sum_{i\in V} e_i^2/s_i = \langle \vec{e}, \vec{e}
    \rangle_{\genS}$
  \end{enumerate}
\end{lemma}
\begin{definition}\label{def:Etild}
    With $f_{ij}(x)$ the expected flow over edge $(i,j)$ 
    in state $x$, we define the set of \emph{non-Nash edges} as
    \begin{equation*}
        \tilde E(x) := \left\{ 
        (i,j) \in E: f_{ij}(x) > 0
        \right\}.
   \end{equation*}
    This is the set of edges for which 
    tasks have an incentive to migrate. Edges with $f_{ij}(x) = 0$ are
    called \emph{Nash edges} or \emph{balanced edges}.
\end{definition}
\begin{definition}
  \label{def:Lambda}
    As an auxiliary quantity, we define
    \begin{equation*}
        \Lambda_{ij}^r(x) := (2\alpha - 2) \cdot d_{ij} \cdot \left(
        \frac{1}{s_i} + \frac{1}{s_j} \right) \cdot
        f_{ij}(x) + \frac{r}{s_i} - \frac{r}{s_j}.
  \end{equation*}
\end{definition}

Our improved bound builds upon results in \cite{Berenbrink2011}. In
that paper, the randomized process is analyzed by first lower-bounding
the potential drop in the case that exactly the expected number of
tasks is moved, and then by upper-bounding the variance of that
process. This leads to Lemma \ref{thm:Lemma35}.  Based on this lemma,
we now prove a stronger bound on the expected drop in the potential
$\Psi_0(x)$. Let us briefly outline the necessary steps.  The lower
bound on the drop in the potential in Lemma \ref{thm:Lemma35} is a sum
over the non-Nash edges and contains terms of the form $\ell_i -
\ell_j$, whereas the potential itself is a sum over the nodes and
contains terms of the form $\ell_i^2$. We will use the graph's
Laplacian matrix to establish a connection between $\Psi_0$ and the
expected drop in $\Psi_0$. This will allow us to prove fast
convergence to a state where $\Psi_0$ is below a certain critical
value $\psi_c$. If $m$ is sufficiently large, this state also is an
\eapx Nash equilibrium.  In the next stage of our approach, we use a
constant drop in $\Psi_1$, a shifted version of $\Phi_1$, to prove
convergence to an exact Nash equilibrium. The techniques from
probability theory used in this this section are similar to the ones
used in \cite{Berenbrink2007a}.

\subsection{Convergence Towards an Approximate Nash Equilibrium}
\label{sec:SelfishANE}
To make the connection with the Laplacian, we first
have to rewrite the bound in Lemma \ref{thm:Lemma35} in the following way.
\begin{lemma}
    \label{thm:DropQuadratic}
    Under the condition that the system is in state $x$, the expected drop in
    the potentials $\Phi_0$ and $\Psi_0$ is bounded by
    \begin{equation*}
      \EV[\Delta \Psi_0(X^{k+1})|X^k = x] \ge
      \sum_{(i,j)\in E} \left[
        \frac{\left(1-\frac{2}{\alpha}\right) \cdot (\ell_i(x) - 
          \ell_j(x))^2}{\alpha
          \cdot d_{i,j}\cdot \left(
            \frac{1}{s_i} + \frac{1}{s_j}\right)}\right] - 
      \frac{n}{\alpha}.
    \end{equation*}
\end{lemma}
Next, we use various technical results from spectral graph theory 
to prove the following bound.
\begin{lemma}
    \label{thm:Psi0DropLambda2}
    Let $L$ be the Laplacian of the network. Let $\lambda_2$ be its second
    smallest eigenvalue. Then
    \begin{equation*}
        \EV[\Delta\Psi_0(X^{k+1})|X^k = x] \ge 
        \frac{\lambda_2}{16 \Delta}\cdot \frac{1}{\smax^2} \cdot 
        \Psi_0 - \frac{n}{4\cdot \smax}.
    \end{equation*}
\end{lemma}

In a first step, we get rid of the conditioning of the potential drop on the
previous state.
\begin{lemma}\label{thm:Psi0DropConditioning}
Let $\gamma$ be defined such that
$1/\gamma = \lambda_2/(32 \Delta \cdot \smax^2)$.

Then, the expected value of the potential in time step $t$ is at most
\begin{equation*}
  \EV[\Psi_0(X^t)] \leq \left(1 - \frac{2}{\gamma}\right)\cdot
  \EV[\Psi_0(X^{t-1})] + \frac{n}{4\cdot\smax}.
\end{equation*}
\end{lemma}

As long as the expected value of the potential is sufficiently large, we can
rewrite the potential drop as a multiplicative drop.
\begin{definition}
  Let $\lambda_2$ be the second smallest eigenvalue of the Laplacian $L(G)$ of
  the network. We define the \emph{critical value} $\psi_c$ as
  $\psi_c = 8\cdot n \cdot \Delta \cdot \smax/\lambda_2$.
\end{definition}
\begin{lemma}
\label{thm:Psi0CriticalFactor}
Let $t$ be a time step for which the expected value of the potential satisfies
$\EV[\Psi_0(X^t)] \ge \psi_c$.
Let $\gamma$ be defined as in Lemma \ref{thm:Psi0DropConditioning}.
Then, the expected potential in time step $t+1$ is bounded by
\begin{equation*}
  \EV[\Psi_0(X^{t+1})] \leq
  \left(1 - \frac{1}{\gamma}\right)\cdot
  \EV[\Psi_0(X^t)].
\end{equation*}
\end{lemma}
This immediately allows us to prove the following.
\begin{lemma}
  \label{thm:Psi0FactorDrop}
  For a given time step $T$, there either is a $t < T$ so that 
  $\EV[\Psi_0(X^t)] \leq \psi_c$, or
  \begin{equation*}
    \EV[\Psi_0(X^T)] \leq \left(1 - \frac{1}{\gamma}\right)^T \cdot 
    \EV[\Psi_0(X^0)].
  \end{equation*}
\end{lemma}
Thus, as long as $\EV[\Psi_0(X^t)] > \psi_c$ holds, the expected
potential drops by a constant factor. This allows us to derive a bound 
on the time it takes until $\EV[\Psi_0(X^t)]$ is small.
\begin{lemma}
    \label{thm:T}
    Let 
    $T = 2\gamma \cdot \ln (m/n)$.
 Then it holds
   \begin{enumerate}
    \item[(1)] There is a $t \leq T$
      such that $\EV[\Psi_0(X^T)] \leq \psi_c$.
    \item[(2)] There is a $t \leq T$ such that the probability 
      that $\Psi_0(X^t) \leq 4\cdot \psi_c$ is at least
      \begin{equation*}
        \PR[\Psi_0(X^t) \leq 4\cdot \psi_c] \ge
        \frac{3}{4}.
      \end{equation*}
    \end{enumerate}
\end{lemma}
This is similar to a result in \cite{Berenbrink2011}, but our factor
$\gamma$ is different. This is reflected in a different expected time needed to
reach an \eapx Nash equilibrium, as we have pointed out in the introduction.

Next, we show that states with $\Psi_0(x) \leq 4\cdot\psi_c$ are indeed
\eapx Nash equilibria if the number of tasks exceeds a certain threshold.
This requires one further observation.
\begin{observation}
    \label{thm:LDeltaPsi0}
    For any state $x$, we have
$L_\Delta(x)^2 \leq \Psi_0(x) \leq \SSum\cdot L_\Delta(x)^2$.
\end{observation}
\begin{lemma}
    \label{thm:psic_means_NE}
    Let $m \ge 8\cdot\delta \cdot n^2 \cdot \SSum \cdot \smax$ 
    for some $\delta > 1$. Then a state $x$ with
    $\Psi_0(x) \leq 4\cdot\psi_c$ is a $2/(1+\delta)$-approximate Nash
    equilibrium.
\end{lemma}
\begin{remark}
  If $m$ is small, it still holds that we reach a state $x$ with $\Psi_0(x) \leq
  4 \cdot \psi_c$, which is all we need to prove convergence to an exact Nash
  equilibrium in the next section. It is just that this intermediate state is
  then not an \eapx-Nash equilibrium.
\end{remark}

Now we are ready to show Theorem \ref{thm:ApproxNE}.
\begin{proof}[Theorem \ref{thm:ApproxNE}]
  Lemma \ref{thm:psic_means_NE} ensures that after $T$ steps the probability for
  \emph{not} having reached a state $x$ with $\Psi_0(x) \leq 4\cdot\psi_c$ is at
  most $1/4$. Hence, the expected number of times we have to repeat $T$ steps is
  less than
  \begin{equation*}
    1 + 1/4 + 1/4^2 + \cdots = \frac{1}{1 - \frac{1}{4}} < 2.
  \end{equation*}
  The expected time needed to reach such a state is therefore at most $2\cdot T$
  with $T$ from Lemma \ref{thm:T}.
\end{proof}
If we let the algorithm iterate until a state $x$ with $\Psi_0(x) \leq 
4\cdot \psi_c$
is obtained, Theorem \ref{thm:ApproxNE} bounds the \emph{expected} number of
time steps we have to perform. However, by repeating a sufficient number
of blocks with $T$ steps, we can obtain arbitrary high probability.
\begin{corollary}
    \label{thm:ProbApproxNE}
    After $c\cdot \log_4 n$ many blocks of size $T$, a state with
    $\Psi_0(x) \leq 4\cdot \psi_c$ is reached with probability
    at least
    $1 - 1/n^c$.
\end{corollary}
\begin{proof}[Corollary~\ref{thm:ProbApproxNE}]
    The probability for \emph{not} reaching
    a state $x$ with $\Psi_0(x) \leq 4\cdot \psi_c$ after $k$ steps
    is at most $1/4^k$. We are interested in the complementary event, so its
    probability is at least $1 - 1/4^k$. For $k = c\cdot \log_4 n$ the statement
    follows immediately.
\end{proof}

\subsection{Convergence Towards a Nash Equilibrium}
\label{sec:NE}
We now prove the upper bound for the expected time necessary to reach an exact
Nash Equilibrium (Theorem \ref{thm:TNE}, p.~\pageref{thm:TNE}).  To show this
result, we have to impose a certain condition on the speeds. If the speeds are
arbitrary non-integers, convergence can become arbitrarily slow. Therefore, we
assume that there exists a common factor $\epsilon \in (0,1]$ so that for every
speed $s_i$ there exists an integer $n_i \in \mathbb{N}$ so that $s_i = n_i\cdot
\epsilon$. We call $\epsilon$ the \emph{granularity} of the speed distribution.
The convergence factor $\alpha$, which was $4\smax$ in the original protocol,
must be changed to $4\smax / \epsilon$. For non-integer speeds, we have
$\epsilon < 1$, so this effectively increases $\alpha$.

To show convergence towards an exact Nash Equilibrium we cannot rely
solely on the potential $\Psi_0(x)$, because when the system is close
to a Nash equilibrium it is possible that the potential function
increases even when a task makes a move that improves its perceived
load.  Therefore, we now look at potential $\Phi_1(x)$.
\begin{definition}
  \label{def:Psi1}
    We define the \emph{shifted potential function}
\begin{equation*}
    \Psi_1(x) = \Phi_1(x) - \frac{m^2}{\SSum} - \frac{m \cdot n}{\SSum}
    - \frac{n^2}{4\SSum} + \frac{1}{4}\sum_i \frac{1}{s_i}.
\end{equation*}
Let $\bar{s}_a$ and $\bar{s}_h$ denote the arithmetic mean and the harmonic mean
of the speeds, i.e.,
$s_a = \sum_{i\in V} s_i/n$ and
$s_h = n / \sum_{i\in V} 1/s_i$.

Then, we can write
\begin{equation*}
    \Psi_1(x) = \Phi_1(x) - \frac{m^2}{\SSum} - \frac{m\cdot n}{\SSum} + 
    \frac{n}{4}\cdot\left( \frac{1}{\bar{s}_h} - \frac{1}{\bar{s}_a} \right).
\end{equation*}
\end{definition}
\begin{observation}
    \label{thm:observationsPsi1}
    The shifted potential $\Psi_1(x)$ has the following properties.
    \begin{enumerate}
        \item[(1)] Let $\vec{e} = \vec{w} - \vec{\bar{w}}$ be the task
          deviation vector. Then
            \begin{equation*}
            \Psi_1(x) = \sum_{i\in V} \left[ \frac{\left( e_i + \frac{1}{2}
            \right)^2}{s_i}\right] - \frac{n}{4\bar{s}_a}.
            \end{equation*}
        \item[(2)] $\displaystyle \Psi_1(x) \ge 0$.
        \item[(3)] $\displaystyle \Psi_1(x) = \Psi_0(x) + \sum_{i\in V}
            \frac{e_i}{s_i} + \frac{n}{4} \cdot \left( \frac{1}{\bar{s}_h} -
            \frac{1}{\bar{s}_a} \right).$
        \item[(4)] $\displaystyle \Delta \Psi_1(X^t) = \Delta \Phi_1(X^t).$
    \end{enumerate}
\end{observation}

Before we can lower-bound the expected drop in $\Psi_1(x)$, we need a technical
lemma regarding a lower bound to the load difference. It is similar to 
\cite[Lemma 3.7]{Berenbrink2011}, which concerned integer speeds,
so the result here is more general.
\begin{lemma}
    \label{thm:liljtilde_tight}
    Every edge $(i,j)$ with $\ell_i - \ell_j > 1/s_j$ also satisfies 
    \begin{equation*}
        \ell_i - \ell_j \ge \frac{1}{s_j} + \frac{\epsilon}{s_i \cdot s_j}.
    \end{equation*}
\end{lemma}
Potential $\Psi_1$ differs from potential $\Psi'$ defined in
\cite{Berenbrink2011} by a constant only. Therefore, potential differences are
the same for both potentials and we can apply results for $\Psi'$ to $\Psi_1$.

\begin{lemma}
    \label{thm:Psi1DropBound}
If the system is in a state $x$ that is not a Nash equilibrium,
then
\begin{equation*}
  \EV[\Delta \Psi_1(X^{k+1})|X^k=x] \ge \frac{\epsilon^2}{8 \Delta \cdot \smax^3}
\end{equation*}
\end{lemma}
Since the results of the previous section apply to $\Psi_0$ whereas now we work 
with $\Psi_1$, we add this technical lemma relating the two.
\begin{lemma}
\label{thm:Psi1Psi0Relation}
    For any state $x$ it holds
    \begin{equation*}
        \Psi_1(x) \leq \Psi_0(x) + \sqrt{\Psi_0(x)\cdot \frac{n}{\bar{s}_h}} + 
        \frac{n}{4}\cdot \left( \frac{1}{\bar{s}_h} - \frac{1}{\bar{s}_a}
        \right).
    \end{equation*}
\end{lemma}
To obtain a bound on the expected time the system needs to reach the NE, we
use a standard argument from martingale theory.
Let us abbreviate 
$V := \epsilon^2/(8\Delta\cdot\smax^3)$.
We introduce a new random variable $Z_t$ which we define as 
$Z_t = \Psi_1(X^t) + t\cdot V$.
\begin{lemma}\label{thm:Supermartingale}
    Let $T$ be the first time step for which the system is in a Nash
    equilibrium. Then, for all times $t \leq T$ we have
    \begin{enumerate}
        \item[(1)] $\displaystyle
            \EV[Z_t | Z_{t-1} = z] \leq z$
        \item[(2)] $\displaystyle
            \EV[Z_t] \leq \EV[Z_{t-1}].$
    \end{enumerate}
\end{lemma}
\begin{corollary}\label{thm:martingale}
Let $T$ be the first time step for which the system is in a Nash equilibrium.
Let $t \wedge T$ be defined as $\min \{t, T\}$. Then the random variable 
$Z_{t \wedge T}$ is a super-martingale.
\end{corollary}
\begin{corollary}\label{thm:optionalstopping}
    Let $T$ be the first time step for which the system is in a Nash
    equilibrium. Then
    $\EV[Z_T] \leq Z_0 = \Psi_1(X^0)$.
\end{corollary}

Now we are ready to show Theorem \ref{thm:TNE}.
\begin{proof}[Theorem \ref{thm:TNE}]
    First, we assume that at time $t = 0$ the system is in a state 
    with $\EV[\Psi_0(X^T)] \leq 4\cdot \psi_c$. 
    Using the non-negativity of
    $\Psi_1(x)$ (Observation \ref{thm:observationsPsi1}) allows us to state
    \begin{align*}
          V\cdot \EV[T] &\leq \EV[\Psi_1(X^T)] + V\cdot \EV[T] = \EV[Z_T]\\
          (\text{\footnotesize Cor.~\ref{thm:optionalstopping}})\quad
          & \leq \EV[Z_0] = \Psi_1(X^0)\\
          (\text{\footnotesize Lem.~\ref{thm:Psi1Psi0Relation}})\quad
          &\leq \Psi_0(X^0) + \sqrt{\Psi_0(X^0) \cdot
            \frac{n}{\bar{s}_h}} + \frac{n}{4} \cdot \left( \frac{1}{\bar{s}_h} -
            \frac{1}{\bar{s}_a} \right) \\
          &\leq 4\cdot\psi_c + \sqrt{4\cdot \psi_c \cdot \frac{n}{\bar{s}_h}} 
          + \frac{n}{4}
    \end{align*}
    Inserting the definition of $\psi_c$ and dividing by $V$ yields
    \begin{equation*}
        \begin{aligned}
            \EV[T] &\leq 8\Delta \cdot \frac{\smax^3}{\epsilon^2} \cdot
            \left[ 4\cdot\frac{16 \cdot n \cdot \Delta \cdot
                \smax}{\lambda_2}
            + \sqrt{4\cdot \frac{16 \cdot n \cdot \Delta \cdot \smax}{\lambda_2}
             \cdot n} + \frac{n}{4}\right] \\
         (\text{\footnotesize Lem.~\ref{thm:lambda2_Delta}})\quad   
         &\leq 
            8\Delta \cdot \frac{\smax^3}{\epsilon^2} \cdot \left[ 
            \frac{64\cdot n \cdot \Delta \cdot \smax}{\lambda_2}
            + \sqrt{32\cdot n^2 \cdot
            \smax\cdot }\cdot \frac{2\Delta}{\lambda_2}
            + \frac{n}{4}
            \right] \\
            &\leq 
            512\cdot \Delta^2 \cdot \frac{\smax^4}{\epsilon^2}\cdot 
            \frac{n}{\lambda_2}
            + 91\cdot \Delta^2 \cdot \frac{\smax^4}{\epsilon^2} \cdot
            \frac{n}{\lambda_2} 
            + 4 \cdot \Delta^2 \cdot \frac{\smax^4}{\epsilon^2} \cdot 
            \frac{n}{\lambda_2}\\
            &= 607\cdot 
            \Delta^2 \cdot \frac{\smax^4}{\epsilon^2} \cdot \frac{n}{\lambda_2}.
        \end{aligned}
    \end{equation*}
    where we have used that $2\Delta/\lambda_2 \ge 1$
    (Lemma~\ref{thm:lambda2_Delta}) to pull that expression outside of the
    square root in the first line.
    
    This bound was derived under the assumption that at $t = 0$ we had a state
    with $\EV[\Psi_0(X^t)] = \leq 4\cdot \psi_c$. If this is not the case, let 
    $\tau$ denote the number of time steps to reach such a state, and let 
    $T'$ denote the additional number of time steps to reach a NE from there. 
    Combining the result from above with Theorem \ref{thm:ApproxNE} allows 
    us to write 
    \begin{equation*}
      \EV[T] = \EV[\tau + T'] = 
      \mathcal{O} \left(\frac{n}{\lambda_2} \cdot \Delta^2 \cdot
      \frac{\smax^4}{\epsilon^2}\right).
    \end{equation*}
\end{proof}
\begin{corollary}
    Similarly to Corollary \ref{thm:ProbApproxNE}, after $c\cdot \log_4 n$
    blocks of $T$ steps we have reached a Nash Equilibrium with probability at
    least 
    $1 - 1/n^c$.
\end{corollary}
\begin{observation}\label{obs1}
Our bound in Theorem \ref{thm:TNE} is asymptotically lower than 
the corresponding bound in \cite{Berenbrink2011}
by at least a factor of
$\Omega\left( \Delta \cdot \diam(G) \right)$.
\end{observation}
\begin{proof}
  Lemma \ref{thm:lambda2_diambound} yields $n \cdot \diam(G) \ge
  4/\lambda_2$.  Additionally, we have $\SSum \ge \smax$, since
  $\smax$ occurs (at least once) in the sum of all speeds. Hence, the
  asymptotic bound from \cite{Berenbrink2011} is larger than
\begin{equation*}
  \mathcal{O} \left(
    n \cdot \frac{\Delta^2}{\lambda_2} \cdot \smax^4 \cdot
    \left[\Delta \cdot \diam(G) \right]
    \right).
\end{equation*}
The first part of this expression is the bound of Theorem \ref{thm:TNE}, so the
expression in the square brackets is the additional factor of the bound from
\cite{Berenbrink2011}.
\end{proof}

\section{Weighted Tasks}\label{sec:weighted}

The set of tasks assigned to
node $i$ is called $x(i)$. The weight of node $i$ becomes 
$W_i(x) = \sum_{\ell \in x(i)} w_\ell$
whereas the corresponding load is defined as 
$\ell_i(x) = W_i(x)/s_i$.

We present a protocol for weighted tasks that differs from the one described in
\cite{Berenbrink2011}. It is presented in Algorithm \ref{alg:DSL_weighted}

\begin{algorithm}[htbp]
  \DontPrintSemicolon
\Begin{%
\ForEach{task $\ell$ in parallel}{%
Let $i = i(l)$ be the current machine of task $\ell$\;
Choose a neighboring machine $j$ uniformly at random\;
\If{$\ell_i - \ell_j > \frac{1}{s_j}$}{%
  Move task $\ell$ from node $i$ to node $j$ with probability
  \[ p_{ij}:= \frac{\mathrm{deg}(i)}{d_{i,j}} \cdot \frac{W_i - W_j}{2\alpha
    \cdot W_i} \] } } }
\caption{Distributed Selfish Load Balancing for weighted
  tasks\label{alg:DSL_weighted}} 
\end{algorithm}

The notable difference to the scheme in \cite{Berenbrink2011} is that in our
case, the decision of a task $\ell$ to migrate or not does not depend on that
task's weight. In the original protocol, a load difference of more than $w_\ell
/ s_j$ would suffice for task $\ell$ to have an incentive to migrate.  In the
modified protocol, a task will only move if the load difference is at least $1 /
s_j$.  The advantage of this approach is that for an edge $(i,j)$, either all or
none of the tasks on node $i$ have an incentive to migrate. This greatly
simplifies the analysis. We will show that the system rapidly converges to a
state where $\ell_i - \ell_j \leq 1/s_j$ for all edges $(i,j)$.  Such a system
is not necessarily a Nash equilibrium as $\ell_i - \ell_j$ might still be larger
than the size of a given task $w_\ell$.  We will show, however, that such a
state is an \eapx NE.

\begin{definition}\label{def:fij_weight}
In analogy to the unweighted case, we define the expected flow $f_{ij}$ as the 
expected \emph{weight} of the tasks migrating from $i$ to $j$ in state $x$. 
It is given by
\begin{equation*}
  f_{ij}(x) = \frac{\ell_i(x) - \ell_j(x)}{\alpha \cdot d_{ij} \cdot 
    \left(\frac{1}{s_i} + \frac{1}{s_j}\right) \cdot W_i(x)} \cdot 
  \sum_{\ell\in x(i)} w_\ell = \frac{\ell_i(x) - \ell_j(x)}{\alpha \cdot d_{ij}
  \cdot \left( \frac{1}{s_i} + \frac{1}{s_j}\right)}.
\end{equation*}
\end{definition}

The potentials $\Phi_0$ and $\Phi_1$ are defined analogously to the unweighted
case. Here, we concentrate on $\Phi_0$ alone. 
The average weight per node is $W/n$ and the task
deviation $e_i$ is defined as $W_i - W/n$. We define $\Psi_0(x)$ in analogy to
the unweighted case as the normalized version of $\Phi_0$,
\begin{equation*}
  \Psi_0 = \Phi_0 - \frac{W^2}{\sum_i s_i} = \sum_{i\in V} \frac{e_i^2}{s_i}.
\end{equation*}
The auxiliary quantity $\Lambda_{ij}(x)$ is defined analogously to the
unweighted case as 
\begin{equation*}
  \Lambda_{ij}(x) = (2\alpha - 2) \cdot d_{ij} \cdot 
  \left(\frac{1}{s_i} + \frac{1}{s_j}\right) \cdot f_{ij}(x).
\end{equation*}

\subsection{Convergence Towards an Approximate Nash Equilibrium}
In close analogy to \cite[Lemma 3.1]{Berenbrink2011}, we first bound the drop of
the potential when the flow is exactly the expected flow.
\begin{lemma}\label{thm:ExpectedDrop_weight}
  The drop in potential $\Phi_0$ if the system is in
  state $x$ and if the flow is exactly the expected flow is bounded by
  \begin{equation*}
      \tilde \Delta \Phi_0(X^{t+1}|X^t = x)
      \ge \sum_{(i,j)\in \tilde E} f_{ij} \cdot \Lambda_{ij}.
  \end{equation*}
\end{lemma}
The proof is formally equivalent to the one in \cite{Berenbrink2011} and
therefore omitted here.
Next, we bound the variance of the process.
\begin{lemma}\label{thm:VarianceWeighted}
The variances of the weights on the nodes are bounded via
\begin{equation*}
  \sum_i \frac{\VAR[W_i(X^t)|X^{t-1}=x]}{s_i} \leq
  \sum_{ij} f_{ij} \cdot \left( \frac{1}{s_i} + \frac{1}{s_j} \right)
\end{equation*}
\end{lemma}
This allows us to formulate a bound on the expected potential drop in analogy to
\cite[Lemma 3.3]{Berenbrink2011} by combining Lemma
\ref{thm:ExpectedDrop_weight} and Lemma \ref{thm:VarianceWeighted}.
\begin{lemma}\label{thm:Psi0Drop_weight}
The expected drop in potential $\Phi_0$ if the system is in state $x$ is at
least   
  \begin{equation*}
 E[\Delta \Phi_0(X^t)|X^{t-1}=x] \ge 
  \sum_{(i,j)\in \tilde E} f_{ij}(x)\cdot(\Lambda_{ij}(x) - 2).
\end{equation*}
\end{lemma}
The proof is analogous to the corresponding lemma in \cite{Berenbrink2011}.

\begin{proof}[Theorem \ref{thm:ANE_weight}]

The rest of the proof is the same as the proof for the unweighted case. One
may verify that, indeed, Lemma \ref{thm:DropQuadratic} and all subsequent
results do not rely on the specific form of $\ell_i$ or the underlying nature of
the tasks.
Using the same eigenvalue techniques as in the unweighted case, this allows us
to obtain a bound involving the second smallest eigenvalue of the graph's
Laplacian matrix.
Following the steps of the unweighted case allows
us to prove the main result of this section.
\end{proof}


 \subsection*{Acknowledgements}
 The authors thank Thomas Sauerwald for helpful discussions.
\bibliography{bibliography,manual}
\appendix

\section{Spectral Graph Theory}
\label{sec:spectral}
In this appendix, we will briefly summarize some important theorems of
spectral graph theory.
For an excellent introduction, we recommend the
book by Fan Chung \cite{Chung1997}. Many important
results are collected in an overview article by Mohar 
\cite{mohar}.

Results in
this section are, unless indicated otherwise, taken from these sources.
Let us begin by defining the matrix we are interested in.
\begin{definition}\label{def:Laplacian}
    Let $G = (V,E)$ be an undirected graph with vertices $V = \{1, \dots n\}$
    and edges $E$.

    The \emph{Laplacian} $L(G)$ of $G$ is defined as
    \begin{align*}
      L(G) & \in \mathbb{N}^{n\times n} &
      L(G)_{ij} = \begin{cases}
        \degree(i) & i = j \\
        -1         & (i,j) \in E \\
        0          & \text{otherwise.}
      \end{cases}
    \end{align*}
\end{definition}
The following Lemma summarizes some basic properties of $\tilde L(G)$ and,
therefore, also of
$L(G)$. These properties are found in every introduction to spectral graph 
theory.
\begin{lemma}
    \label{thm:basic_L}
    Let $L(G)$ be the Laplacian of a graph $G$. For brevity, we
    omit the argument $G$ in the following. Then, $L$ satisfies the
    following.
    \begin{itemize}
        \item[(1)] For every vector $\vec{x} \in \mathbb{R}^n$ we have 
            \begin{equation*}
                \vec{x}^\top L \vec{x} = 
                \sum_{i,j \in V} x_i \cdot L_{ij} \cdot x_j =
                \sum_{(i,j)\in E} c_{ij} \cdot (x_i - x_j)^2
           \end{equation*}
        \item[(2)] $L$ is symmetric positive semi-definite, i.e.,
          $L^\top = L$ and $\vec{x}^\top L \vec{x} \ge 0$
          for every vector $\vec{x}$.
        \item[(3)] Each column (row) of $\tilde L$ sums to $0$.
    \end{itemize}
\end{lemma}

\subsection{Spectral Analysis}
We now turn our attention to the spectrum of the Laplacian. 
\begin{definition}
  \label{def:LaplacianSpectrum}
  Let $L(G)$ be the Laplacian of a graph $G$.  Lemma \ref{thm:basic_L} and the
  spectral theorem of linear algebra ensure that $L$ has an orthogonal
  eigenbasis, i.e. there are $n$ (not necessarily distinct) eigenvalues with $n$
  linearly independent eigenvectors which can be chosen to be mutually
  orthogonal.

  We call the eigenvalues of $L(G)$ the \emph{Laplacian spectrum} of $G$ and
  write
  \begin{equation*}
    \lambda(G) = (\lambda_1 \leq \lambda_2 \cdots \le \lambda_n)
  \end{equation*}
  where the $\lambda_i$ are the eigenvalues of $L(G)$.

  The corresponding eigenvectors are denoted $\vec{v}_i$.
\end{definition}
The Laplacian spectrum of $G$ contains valuable information about $G$. Some
very basic results are given in the next Lemma.
\begin{lemma}
    \label{thm:smallestEV}
    Let $G$ be a graph with Laplacian spectrum $\lambda(G)$.  For a graph $G =
    (V,E)$ the following holds for both the unweighted and the weighted
    spectrum.
    \begin{itemize}
        \item[(1)] The vector $\vec{1} := (1, \cdots, 1)^\top$ is eigenvector
            to $L$ and $\tilde L$ with eigenvalue $0$. Hence, $\lambda_1 = 0$
            is always the smallest eigenvalue of any Laplacian.
        \item[(2)] The multiplicity of the eigenvalue $0$ is equal to the
            number of connected components of $G$. In particular, a connected
            graph has $\lambda_1 = 0$ and $\lambda_2 > 0$.
    \end{itemize}
\end{lemma}
The second-smallest eigenvalue $\lambda_2$ is closely
related to the connectivity properties of $G$. It was therefore called \emph{algebraic connectivity}
when it was first intensely studied by Fiedler \cite{Fiedler}. The eigenvector
corresponding to $\lambda_2$ is also called \emph{Fiedler vector}.
A first, albeit weak, result is the preceding lemma. A stronger result with
a corollary useful for simple estimates is given in the next lemma.
\begin{lemma}[\cite{Mohar1991}]
  \label{thm:lambda2_diambound}
    Let $\lambda_2$ be the second-smallest eigenvalue of the unweighted
    Laplacian of a graph $G$. Let $\diam(G)$ be the diameter of graph $G$. Then
    \begin{equation*}
        \diam(G) \ge \frac{4}{n\cdot \lambda_2}.
   \end{equation*}
\end{lemma}
\begin{corollary}
    \label{thm:lambda2_simplebound}
    Using $\diam(G) \leq n$, we get
    $\displaystyle \lambda_2 \ge \frac{4}{n^2}$
\end{corollary}
\begin{lemma}
    \label{thm:lambda2_Delta}
This is another useful result by Fiedler \cite{Fiedler}. Let $\lambda_2$ be the
second-smallest eigenvalue of $L(G)$. Then,
\begin{equation*}
    \lambda_2 \leq \frac{n}{n-1}\cdot \min \{ \deg(i), i \in V\}.
\end{equation*}
For $\Delta$ the maximum degree of graph $G$, it immediately follows
\begin{equation*}
    \lambda_2 \leq \frac{n}{n-1}\cdot \Delta.
\end{equation*}
\end{lemma}

A stronger relationship between $\lambda_2$ and the network's connectivity
properties is provided via the graph's Cheeger constant.
\begin{definition}
  Let $G = (V,E)$ be a graph and $S \subset V$ a subset of the nodes. The
  \emph{boundary} $\delta S$ of $S$ is defined as the set of edges having
  exactly one endpoint in $S$, i.e.,
  \begin{equation*}
    \delta S = 
    \{ (i,j) \in E \mid i \in S, j \in V \setminus S \}.
 \end{equation*}
\end{definition}
\begin{definition}
  Let $G = (V,E)$ be a graph.
            The \emph{isoperimetric number} $i(G)$ of $G$ is defined as 
            \begin{equation*}
                i(G) = \min_{\substack{S \subset V\\ |S| \leq |V|/2}} 
                \frac{|\delta S|}{|S|}.
           \end{equation*}
            It is also called \emph{Cheeger constant} of the graph.
\end{definition}
The isoperimetric number of a graph is a measure of how well any subset of the
graph is connected to the rest of the graph. Graphs with a high Cheeger
constant are also called \emph{expanders}. The following was proven by
Mohar.
\begin{lemma}[\cite{MOHAR1989}]
    \label{thm:CheegerBound}
    Let $\lambda_2$ be the second-smallest eigenvalue of $L(G)$, and let $i(G)$
    be the isoperimetric number of $G$. Then,
    \begin{equation*}
        \frac{i^2(G)}{2 \Delta} \leq \lambda_2 \leq 2 i(G).
   \end{equation*}
\end{lemma}

This concludes our introduction to spectral graph theory, which suffices for the
analysis of identical machines. For machines with speeds, it turns out that a
generalized Laplacian is a more expressive quantity. 

\subsection{Generalized Laplacian Analysis}
\label{sec:genlap}
Recall the speed-matrix $S$ from the introduction. Instead of
analyzing the Laplacian $L$, we are now interested in the
\emph{generalized Laplacian}, defined as $LS^{-1}$. This definition is
also used by Els\"asser in \cite{Elsasser2002} in the analysis of
continuous diffusive load balancing in heterogeneous networks. In this
reference, the authors prove a variety of results for the generalized
Laplacian, which we restate here in a slightly different language.

It turns out that in the discussion of the properties of
this generalized Laplacian, many results carry over from the analysis of the
normal Laplacian. The similarity is made manifest by the introduction of a 
\emph{generalized dot-product}.
\begin{definition}
    \label{def:dotS}
    For vectors $\vec{x}, \vec{y} \in \mathbb{R}^{n}$, we define the
    \emph{generalized dot-product} with respect to $S$ as 
    \begin{equation*}
        \langle \vec{x}, \vec{y} \rangle_{\genS} := \vec{x}^T S^{-1} \vec{y} =
        \sum_{i\in V} \frac{x_i\cdot y_i}{s_i}
    \end{equation*}
\end{definition}
\begin{lemma}
    The vector space $\mathbb{R}^n$ together with $\langle \cdot, \cdot
    \rangle_{\genS}$ forms an \emph{inner product space}. This means that
    \begin{itemize}
    \item[(1)] $\langle \vec{x}, \vec{y} \rangle_{\genS} = \langle
      \vec{y}, \vec{x}\rangle_{\genS}$, i.e., $\langle \cdot, \cdot
      \rangle_{\genS}$ is symmetric,
    \item[(2)] $\langle a\vec{x}_1 + b\vec{x}_2, y\rangle_{\genS} =
      a\langle \vec{x}_1, \vec{y}\rangle_{\genS} + b\langle \vec{x}_2,
      \vec{y}\rangle_{\genS}$ for any scalars $a$ and $b$, i.e.,
      $\langle \cdot, \cdot \rangle_{\genS}$ is linear in its first
      argument,
    \item[(3)] $\langle \vec{x}, \vec{x} \rangle_{\genS} \ge 0$, with
      equality if and only if $\vec{x} = 0$, i.e., $\langle \cdot,
      \cdot \rangle_{\genS}$ is positive definite.
    \end{itemize}
\end{lemma}
\begin{proof}
  All three properties follow immediately from Definition \ref{def:dotS},
  provided the $s_i$ are positive, which is true in our case.
\end{proof}
\begin{remark}
    The fact that $\langle \cdot, \cdot \rangle_{\genS}$ is an \emph{inner product}
    allows us to directly apply many results of linear algebra to it. 
    For example, all inner products satisfy the Cauchy-Schwarz inequality,
    i.e.,
    \begin{equation*}
        \langle \vec{x}, \vec{y} \rangle_{\genS}^2 \leq 
        \langle \vec{x}, \vec{x} \rangle_{\genS} \cdot
        \langle \vec{y}, \vec{y} \rangle_{\genS}.
   \end{equation*}
    A proof of this important inequality can be found in every introductory
    book on Linear Algebra.

    \label{def:OrthoWRTS}
    Another concept is that of orthogonality. Two vectors $\vec{x}$
    and $\vec{y}$ are called orthogonal to each other, $\vec{x} \bot \vec{y}$,
    if $\vec{x} \cdot \vec{y} = 0$. Analogously, we call $\vec{x}$ and
    $\vec{y}$ orthogonal \emph{with respect to} $S$ if $\langle \vec{x}, \vec{y}
    \rangle_{\genS} = 0$.
\end{remark}
Let us now collect some of the properties of $LS^{-1}$. These properties have
also been used in \cite{Elsasser2002}. We restate them here using the notation
of the generalized dot product.
\begin{lemma}
  (Compare Lemma 1 in \cite{Elsasser2002})
    \label{thm:GenLap}
    Let $L$ be the Laplacian of a graph, and let $S$ be the speed-matrix,
    $S = \diag(s_1, \cdots, s_n)$. Then the following holds true for the
    generalized Laplacian
    $LS^{-1}$.
    \begin{itemize}
    \item[(1)] The speed-vector $\vec{s} = (s_1, \dots, s_n)^\top$ is
      (right-)eigenvector to $LS^{-1}$ with eigenvalue $0$.
    \item[(2)] $LS^{-1}$ is not symmetric any more. It is, however, still
      positive semi-definite.
    \item[(3)] Since $LS^{-1}$ is not symmetric, we have to distinguish left-
      and right-eigenvectors. Similar to the spectral theorem of linear algebra,
      we can find a basis of right-eigenvectors of $LS^{-1}$ that are orthogonal
      with respect to $S$.
    \end{itemize}
\end{lemma}
\begin{proof}
    (1)
    \begin{equation*}
        LS^{-1}\vec{s} = L\vec{1} = 0
    \end{equation*}
    via Lemma \ref{thm:smallestEV}. 
    For (2) and (3), suppose that $\vec{x}$ is a right-eigenvector of
    $LS^{-1}$ with eigenvalue $\lambda$. If we define
    $\vec{y} := S^{-1/2} \vec{x}$, then we have 
    \begin{alignat*}{2}
        & & LS^{-1}\vec{x} &= \lambda \vec{x} \nonumber \\
        \Leftrightarrow & \quad & 
        LS^{-1/2} \vec{y} &= \lambda S^{1/2} \vec{y} \nonumber \\
        \Leftrightarrow & \quad &
        S^{-1/2} L S^{-1/2} \vec{y} &= \lambda \vec{y}.
    \end{alignat*}
    This proves that $\vec{x}$ is right-eigenvector to $LS^{-1}$ with
    eigenvalue $\lambda$ if and only if $S^{-1/2} \vec{x}$ is eigenvector to 
    $S^{-1/2} L S^{-1/2}$ with eigenvalue $\lambda$. The latter matrix is
    positive definite, because for every vector $\vec{x}$, we have 
    \begin{equation*}
        \vec{x}^\top S^{-1/2} L S^{-1/2} \vec{x} = 
        (S^{-1/2} \vec{x})^\top L (S^{-1/2} \vec{x}) \ge 0
    \end{equation*}
    since $L$ itself is positive semi-definite.  Now, since $S^{-1/2} L
    S^{-1/2}$ is symmetric positive semi-definite, all its eigenvalues are real
    and non-negative and it possesses an orthogonal eigenbasis. Let us denote
    the $n$ vectors of the eigenbasis with $\vec{y}^k$, $k = 1 \dots n$. As we
    have just shown, this implies that the vectors $\vec{x}^k = S^{1/2}
    \vec{y}^k$ are right-eigenvectors to $LS^{-1}$. Since $S^{1/2}$ is a matrix
    of full rank, the $\vec{x}^k$ form a basis as well.  Their orthogonality
    with respect to $S$ follows from
    \begin{align*}
        \langle \vec{x}^k, \vec{x}^l \rangle_{\genS} &=
        (\vec{x}^k)^\top S^{-1} \vec{x}^l \nonumber \\
        &=
        (S^{-1/2}\vec{x}^k)^\top  S^{-1/2} \vec{x}^l \nonumber \\
        &= (\vec{y}^k)^\top \cdot \vec{y}^l = 0\quad \text{for}\, k \not= l.
    \end{align*}
\end{proof}
For arbitrary vectors, we know that $\langle \vec{x}, LS^{-1} \vec{x} \rangle_{\genS}
\ge 0$ since $S^{-1} L S^{-1}$ is positive semi-definite.
The next lemma bounds the generalized dot product of certain vectors with the
Laplacian with the second smallest right-eigenvector of it. A similar version 
can also be found in \cite[Section 3]{Elsasser2002}.
\begin{lemma}
    \label{thm:Laplacian_Eigenvalue}
    Let $\lambda_2$ denote the second-smallest right-eigenvalue of the
    generalized Laplacian,
    $LS^{-1}$. Let $\vec{e}$ be a vector that is orthogonal to the speed vector
    with respect to $S$, i.e.
    $\langle \vec{e}, \vec{s} \rangle_{\genS} = 0$.
    Then 
    \begin{equation*}
        \langle \vec{e}, LS^{-1} \vec{e} \rangle_{\genS} \ge \lambda_2 \langle
        \vec{e}, \vec{e} \rangle_{\genS}.
   \end{equation*}
\end{lemma}
\begin{proof}
  Let $(\lambda_k, \vec{v}_k)$ denote the $k$-smallest eigenvalue and
  corresponding eigenvector of $LS^{-1}$.  For the speed vector, $\vec{s}$, we
  have $LS^{-1}\vec{s} = 0$ (Lemma \ref{thm:GenLap}). Thus, we can just identify
  $\vec{v}_1 = \vec{s}$.
    Recall from Lemma \ref{thm:GenLap} that the $\vec{v}_k$ form a basis of
    $\mathbb{R}^n$.  Therefore, $\vec{e}$ can be written as a linear combination
    of these eigenvectors. For some real-valued coefficients $\beta_k$, we have
    \begin{equation*}
      \vec{e} = \sum_{k = 1}^n \beta_k \vec{v}_k,
    \end{equation*}
    Since the basis vectors are mutually orthogonal with respect to $S$, and
    since $\vec{e}$ is orthogonal to $\vec{s}$ with respect to $S$, $\vec{s} =
    \vec{v}_1$ does not contribute to the linear combination of $\vec{e}$,
    because
    \begin{equation*}
      0 = \langle \vec{e}, \vec{v}_1 \rangle_{\genS} = 
      \beta_1 \langle \vec{v}_1, \vec{v}_1 \rangle.
    \end{equation*}
    This can only be satisfied if either all speeds are zero or if $\beta_1 =
    0$. Therefore, we can write
    \begin{equation*}
        \vec{e} = \sum_{k = 2}^n \beta_k \vec{v}_k.
    \end{equation*}
    Substituting this decomposition into 
    the Bound of Lemma \ref{thm:Laplacian_Eigenvalue}
    yields
    \begin{align*}
      \langle \vec{e}, LS^{-1} \sum_{k=2}^n \beta_k \vec{v}_k \rangle_{\genS} &=
      \sum_{k=2}^n \lambda_k \cdot \beta_k \langle \vec{e}, \vec{v}_k \rangle_{\genS}
      \nonumber \\
      &= \sum_{k=2}^n \lambda_k \cdot \beta_k^2 \langle \vec{v}_k, \vec{v}_k
      \rangle_{\genS}
      \nonumber \\
      &\ge \lambda_2 \cdot \sum_{k=2}^{n} \langle \beta_k \vec{v}_k, \beta_k
      \vec{v}_k \rangle_{\genS} = \lambda_2 \cdot \langle \vec{e}, \vec{e} \rangle_{\genS}.
    \end{align*}
\end{proof}
The next technical lemma is needed to relate the spectra of $L$ and
$LS^{-1}$. We require 
this relation because most of the useful results and bounds for
$\lambda_2$ apply to the normal Laplacian only. 
\begin{lemma}
    \label{them:horn_inequality}
    Let $\mu_i$ denote the eigenvalues of $LS^{-1}$ in ascending order and
    let $\lambda_i$ denote the eigenvalues of $L$ in ascending order. Finally,
    let
    $s_i$ denote the speeds in descending order. Then
    \begin{align}
        \mu_{i+j-1} &\ge \frac{\lambda_i}{s_j} \label{eqn:L_LS_lower_bound}
        & \quad & 0 \leq i, j \leq n,\quad 0 \leq i+j -1\leq n\\
        \mu_{i+j-n} &\leq \frac{\lambda_i}{s_j} & \quad &
        0 \leq i,j \leq n,\quad 0 \leq i+j-n \leq n.
        \label{eqn:L_LS_upper_bound} 
    \end{align}
\end{lemma}
   \begin{proof}[Lemma \ref{them:horn_inequality}]
        The matrices $L$ and $S^{-1}$ are symmetric positive semi-definite.
        Hence, their square-roots exist and are unique. Let 
        $X = \sqrt{L}$ and $T = \sqrt{S^{-1}}$. The singular values of 
        $X$ are $\sqrt{\mu_i}$ and those of $T$ are 
        $\sqrt{s_i^{-1}}$. In addition, the singular values of 
        $XT$ are $\sqrt{\lambda_i}$.

        By Theorem \ref{thm:singular_hermition}, there exist symmetric matrices
        $H_1$ with eigenvalues $\log \sqrt{\mu_i}$ and $H_2$ with eigenvalues 
        $\log \sqrt{s_i^{-1}}$, and the eigenvalues of $H_1 + H_2$ are 
        $\log \sqrt{\lambda_i}$. By Theorem \ref{thm:Weyl}, these satisfy the
        inequalities
        \begin{align*}
            \log \sqrt{\lambda_{i+j-1}} &\ge \log\sqrt{\mu_i} + \log
            \sqrt{s_j^{-1}} = \log\sqrt{ \mu_i s_j^{-1}} \\
            \log \sqrt{\lambda_{i+j-n}} &\leq \log\sqrt{\mu_i} + \log
            \sqrt{s_j^{-1}} = \log\sqrt{ \mu_i s_j^{-1}}
        \end{align*}
        Since both the logarithm and the square-root are monotone functions,
        the desired result follows immediately. 
    \end{proof}
\begin{corollary}
    \label{thm:LS_L_Eigenvalues}
    Let $\mu_2$ denote
    the second smallest right eigenvalue of $LS^{-1}$ and let
    $\lambda_2$ denote the second smallest eigenvalue of $L$. 
    Let $\smax = s_1$ be the largest speed and $\smin = s_n$ the smallest speed.
    Then
    \begin{equation*}
        \frac{\lambda_2}{\smax} \leq \mu_2 \leq \frac{\lambda_2}{\smin}.
    \end{equation*}
\end{corollary}
\begin{proof}
    Let $i = 2, j = 1$ in (\ref{eqn:L_LS_lower_bound}) and 
    $i = 2, j = n$ in (\ref{eqn:L_LS_upper_bound}).
\end{proof}

  \section{Proofs from Section~\ref{sec:selfish}}
  \begin{proof}[Lemma \ref{thm:Psi0properties}]
    (1) By definition, $\Psi_0(x)$ and $\Phi_0(x)$ differ by $m^2/\SSum$. The
    total number of tasks, $m$, and the sum of all speeds, $\SSum$, are
    constants. Therefore, the difference between $\Psi_0(x)$ and $\Phi_0(x)$ is
    constant at any time and we have
    \begin{equation*}
      \Delta \Psi_0(X^t | X^{t-1} = x) =
      \Delta \Phi_0(X^t | X^{t-1} = x) - 
      \underbrace{\Delta\left[ \frac{m^2}{\SSum}\right]}_{=0}.
    \end{equation*}
    Hence, both the original and the shifted potential have the same potential
    drop.

    (2) This follows immediately from Definition \ref{def:dotS} of the
    generalized dot-product.
  \end{proof}
  \subsection{Proofs from Section \ref{sec:SelfishANE}}
  \label{sec:proofsSelfishANE}
  \begin{proof}[Lemma~\ref{thm:DropQuadratic}]
    For brevity, we omit the argument $x$ from all quantities. Note that we can
    look at the drop of either $\Phi_0$ or $\Psi_0$.  Substituting the
    particular forms of $f_{ij}$ and $\Lambda_{ij}^0$ (Definitions \ref{def:fij}
    and \ref{def:Lambda}) into the bound provided by Lemma \ref{thm:Lemma35} of
    Lemma 3.3 in \cite{Berenbrink2011}, we arrive at
    \begin{equation*}
      \EV[\Delta \Psi_0(X^{k+1})|X^k=x] \ge 
      \sum_{(i,j)\in \tilde E}
      \left[
        \frac{\left(2-\frac{2}{\alpha}\right)\cdot
          (\ell_i-\ell_j)^2}{\alpha \cdot d_{i,j}\cdot \left( \frac{1}{s_i} +
            \frac{1}{s_j} \right)}
        -
        \frac{(\ell_i - \ell_j)}{\alpha \cdot d_{i,j}}
      \right].
      \tag{$*$}
      \label{eqn:psidropA}
    \end{equation*}
    We define subsets $\tilde E$, $\tilde E_1$ and $\tilde E_2$ of $E$,
    \begin{align*}
      \tilde E &= \left\{ (i,j) \in E \mid \ell_i - \ell_j
        \ge \frac{1}{s_j}\right\} \\
      \tilde E_1 &= \left\{ (i,j)\in \tilde E \mid \ell_i - \ell_j \ge
        \frac{1}{s_i} + \frac{1}{s_j} \right\} \\
      \tilde E_2 &= \tilde E \setminus \tilde E_1.
    \end{align*}
    Note that $\tilde E = \tilde E_1 \cup \tilde E_2$ and $\tilde E_1 \cap
    \tilde E_2 = \emptyset$. Thus, we can split the sum in (\ref{eqn:psidropA})
    into a sum over $\tilde E_1$ and a sum over $\tilde E_2$. We will now bound
    these sums individually.

    Let $(i,j) \in \tilde E_1$ be an edge in $\tilde E_1$ so that $\ell_i \ge
    \ell_j$. Then the definition of $\tilde E_1$ and the non-negativity of
    $\ell_i - \ell_j$ allows us to deduce
    \begin{equation*}
      \ell_i - \ell_j \ge \frac{1}{s_i} + \frac{1}{s_j}
      \Leftrightarrow
      \frac{1}{\frac{1}{s_i} + \frac{1}{s_j}}
      \cdot (\ell_i -\ell_j)^2 \ge
      \ell_i - \ell_j.
    \end{equation*}
    This allows us to bound
    \begin{align*}
      \sum_{(i,j)\in \tilde E_1} \left[
        \frac{\left(2-\frac{2}{\alpha}\right)\cdot (\ell_i-\ell_j)^2}{\alpha
          \cdot d_{i,j} \cdot \left( \frac{1}{s_i} + \frac{1}{s_j} \right)} -
        \frac{(\ell_i - \ell_j)}{\alpha \cdot d_{i,j}} \right]
      &\ge \sum_{(i,j)\in \tilde E_1} \frac{\left(1-\frac{2}{\alpha}\right)\cdot
        (\ell_i-\ell_j)^2}{\alpha \cdot d_{i,j} \cdot \left( \frac{1}{s_i} +
          \frac{1}{s_j} \right)}.\tag{$*$}
      \label{eqn:dropPsiB}
    \end{align*}

    Next, we turn to $\tilde E_2$ and bound
    \begin{equation*}
      \sum_{(i,j)\in \tilde E_2}
      \left[
        \frac{\left(2-\frac{2}{\alpha}\right)\cdot (\ell_i-\ell_j)^2}{\alpha \cdot
          d_{i,j} \cdot \left( \frac{1}{s_i} +
            \frac{1}{s_j} \right)}
        -
        \frac{(\ell_i - \ell_j)}{\alpha \cdot d_{i,j}}
      \right].
    \end{equation*}
    The sum is over two terms, a positive and a negative one. For the first,
    positive term, we simply bound
    \begin{equation*}
      \sum_{(i,j)\in \tilde E_2} \left[
        \frac{
          \left(2-\frac{2}{\alpha}\right)
          \cdot (\ell_i-\ell_j)^2}{\alpha \cdot
          d_{i,j} \cdot \left( \frac{1}{s_i} +
            \frac{1}{s_j} \right)}
      \right] \ge
      \sum_{(i,j)\in \tilde E_2} \left[
        \frac{
          \left(1-\frac{2}{\alpha}\right)
          \cdot (\ell_i-\ell_j)^2}{\alpha \cdot
          d_{i,j} \cdot \left( \frac{1}{s_i} +
            \frac{1}{s_j} \right)}
      \right].
      \tag{$**$}
      \label{eqn:dropPsiC}
    \end{equation*}
    For the edges in $\tilde E_2$, we have $\ell_i - \ell_j < 1/s_i +
    1/s_j$. This allows us to bound the second, negative term via
    \begin{equation*}
      \sum_{(i,j)\in \tilde E_2} \frac{\ell_i - \ell_j}{\alpha \cdot d_{i,j}} \leq 
      \frac{1}{\alpha}\cdot \sum_{(i,j)\in \tilde E_2} \frac{1}{d_{i,j}}\cdot\left(
        \frac{1}{s_i} + \frac{1}{s_j}
      \right)
      \tag{$***$}
      \label{eqn:dropPsiD}
    \end{equation*}
    Combining (\ref{eqn:dropPsiB}), (\ref{eqn:dropPsiC}) and
    (\ref{eqn:dropPsiD}) yields
    \begin{alignat*}{2}
      \EV[\Delta \Psi_0(X^{k+1})|X^k=x] &\ge & & \sum_{(i,j)\in \tilde E} \left[
        \frac{\left(2-\frac{2}{\alpha}\right)\cdot (\ell_i-\ell_j)^2}{\alpha
          \cdot d_{i,j}\cdot \left( \frac{1}{s_i} + \frac{1}{s_j} \right)} -
        \frac{(\ell_i - \ell_j)}{\alpha \cdot d_{i,j}} \right]
      \\
      &\ge & & \sum_{(i,j)\in \tilde E} \frac{\left(1 -
          \frac{2}{\alpha}\right)\cdot (\ell_i - \ell_j)^2}{\alpha \cdot d_{ij}
        \cdot \left(\frac{1}{s_i} + \frac{1}{s_j}\right)} - \sum_{(i,j)\in
        \tilde E_2} \frac{1}{\alpha\cdot d_{ij}} \cdot \left(\frac{1}{s_i} +
        \frac{1}{s_j}\right).  \tag{$\dagger$}
      \label{eqn:psidropE}
    \end{alignat*}
    In the next step, we rewrite the sum over $\tilde E$ in (\ref{eqn:psidropE})
    to a sum over all edges $E$, using $\tilde E = E \setminus (E\setminus
    \tilde E)$.  It generally holds for any terms $X_{(i,j)}$ that
    \begin{equation*}
      \sum_{(i,j)\in \tilde E} X_{(i,j)} =
      \sum_{(i,j)\in E} X_{(i,j)} -
      \sum_{(i,j) \in E\setminus \tilde E} X_{(i,j)}.
    \end{equation*}
    We will apply this to (\ref{eqn:psidropE}).  In the following, we therefore
    prove an upper bound on the sum over $E \setminus \tilde E$.  Without loss
    of generality, let the nodes $i$ and $j$ of an edge be ordered such that
    $\ell_i \ge \ell_j$.  For edges not in $\tilde E$, we have, by definition,
    $0 \leq \ell_i - \ell_j \leq \frac{1}{s_j}$, so this part can be bound by
    \begin{align*}
      \sum_{(i,j) \in E\setminus \tilde E} \frac{\left( 1 - \frac{2}{\alpha}
        \right)}{\alpha \cdot d_{i,j} \cdot \left( \frac{1}{s_i} + \frac{1}{s_j}
        \right)}\cdot (\ell_i - \ell_j)^2 &\leq \sum_{(i,j)\in E\setminus \tilde
        E} \frac{1}{\alpha\cdot d_{ij}}\cdot
      \frac{s_i}{s_i + s_j} \cdot \frac{1}{s_j} \\
      &\leq \sum_{(i,j)\in E\setminus \tilde E} \frac{1}{\alpha\cdot d_{ij}}
      \cdot
      \left(1 - \frac{s_j}{s_i + s_j}\right) \cdot \frac{1}{s_j}\\
      &= \sum_{(i,j)\in E\setminus\tilde E} \frac{1}{\alpha \cdot d_{i,j}}
      \cdot\left( \frac{1}{s_j} - \frac{1}{s_i + s_j}
      \right) \\
      &\leq \sum_{(i,j)\in E\setminus \tilde E} \frac{1}{\alpha \cdot
        d_{i,j}}\cdot\left( \frac{1}{s_i} + \frac{1}{s_j} \right).
    \end{align*}
    This bound has the same form as the bound in (\ref{eqn:dropPsiD}), only that
    it goes over $E \setminus \tilde E$ instead of $\tilde E_2$.  These two sets
    are disjunct, since $\tilde E_2 \subset \tilde E$.  Therefore, we can
    combine the two sums into a single sum over $\tilde E_2 \cup (E \setminus
    \tilde E) = E \setminus \tilde E_1$.  We then obtain from the following
    bound from (\ref{eqn:psidropE}).
    \begin{alignat*}{2}
      \EV[\Delta\Psi_0(X^{k+1})|X^k=x] &\ge && \sum_{(i,j)\in E} \left[
        \frac{\left(1-\frac{2}{\alpha}\right)\cdot (\ell_i - \ell_j)^2}{\alpha
          \cdot d_{i,j} \cdot \left(
            \frac{1}{s_i} + \frac{1}{s_j} \right)} \right]\\
      &\quad &-{} &\sum_{(i,j) \in E \setminus \tilde E_1} \frac{1}{\alpha\cdot
        d_{ij}} \cdot \left(\frac{1}{s_i} + \frac{1}{s_j}\right).  \tag{$\dagger
        \dagger$}
      \label{eqn:psidropF}
    \end{alignat*}
    The first term of this bound already has the desired form. We will now bound
    the second term. Since it is negative, we have to upper bound the sum
    itself. First, note that $E \setminus \tilde E_1$ is a subset of $E$. As the
    term inside the sum is non-negative, we can write
    \begin{equation*}
      \sum_{(i,j)\in E \setminus \tilde E_1} \frac{1}{\alpha
        \cdot d_{ij}} \cdot \left(\frac{1}{s_i} + \frac{1}{s_j}\right) \leq
      \frac{1}{\alpha}\cdot \sum_{(i,j) \in E} \left[
        \frac{1}{d_{ij} \cdot s_i} + \frac{1}{d_{ij} \cdot s_j}\right].
    \end{equation*}
    Recall that $d_{ij}$ is defined as $\max \{ \degree(i), \degree(j) \}$, so
    we can bound
    \begin{align*}
      \frac{1}{\alpha}\cdot \sum_{(i,j) \in E} \left[ \frac{1}{d_{ij} \cdot s_i}
        + \frac{1}{d_{ij} \cdot s_j}\right] &\leq \frac{1}{\alpha}\cdot
      \sum_{(i,j)\in E} \left[
        \frac{1}{\degree(i) \cdot s_i} + \frac{1}{\degree(j) \cdot s_j}\right]\\
      &= \frac{1}{\alpha} \cdot \sum_{i \in V} \sum_{j \in \adj(i)}
      \frac{1}{\degree(i)\cdot s_i} \\
      &= \frac{1}{\alpha} \cdot \sum_{i \in V} \frac{1}{s_i} \leq
      \frac{n}{\alpha}.
    \end{align*}
    Inserting this bound into (\ref{eqn:psidropF}) yields the result.
  \end{proof}
  \begin{proof}[Lemma~\ref{thm:Psi0DropLambda2}]
    We start from the bound obtained in Lemma \ref{thm:DropQuadratic}.  In the
    course of this proof, we will use Lemma \ref{thm:Laplacian_Eigenvalue} for
    the task deviation vector $\vec{e}$.  In order to use this lemma, we have to
    show that $\vec{e}$ satisfies the lemma's condition, i.e., $\langle \vec{e},
    \vec{s} \rangle_{\genS} = 0$.  This follows via
    \begin{equation*}
      \langle \vec{e}, \vec{s} \rangle_{\genS} = \sum_{i \in V} \frac{e_i\cdot s_i}{s_i}
      = \sum_{i\in V} e_i = 0.
    \end{equation*}
    We can now begin with the main proof.
    \begin{align*}
      \EV[\Delta \Psi_0(X^{k+1})|X^k=x] &\ge \sum_{(i,j)\in E} \left[
        \frac{\left(1-\frac{2}{\alpha}\right) \cdot (\ell_i(x) -
          \ell_j(x))^2}{\alpha \cdot d_{i,j}\cdot \left( \frac{1}{s_i} +
            \frac{1}{s_j}\right)}\right] -
      \frac{n}{\alpha} \\
      &\ge \frac{\frac{1}{2}}{4\cdot \smax \cdot \Delta \cdot 2} \cdot
      \sum_{(i,j)\in E} \left(\ell_i - \frac{m}{\SSum} + \frac{m}{\SSum} -
        \ell_j\right)^2 -\frac{n}{4\cdot \smax}
      \\
      &\ge \frac{1}{16\Delta}\cdot \frac{1}{\smax} \cdot \sum_{(i,j)\in E}
      \left(\frac{e_i}{s_i} - \frac{e_j}{s_j}\right)^2
      - \frac{n}{4\cdot \smax}\\
      (\text{Lem.~\ref{thm:GenLap}}) &= \frac{1}{16\Delta}\cdot
      \frac{1}{\smax}\cdot \left( S^{-1} \vec{e}\right)^\top L \left(S^{-1}
        \vec{e}\right) - \frac{n}{4\cdot \smax}
      \\
      &= \frac{1}{16\Delta}\cdot \frac{1}{\smax}\cdot \langle \vec{e}, LS^{-1}
      \vec{e} \rangle_{\genS}
      - \frac{n}{4\cdot \smax} \\
      \text{(Lem.~\ref{thm:Laplacian_Eigenvalue},
        Cor.~\ref{thm:LS_L_Eigenvalues})} &\ge \frac{1}{16\Delta}\cdot
      \frac{1}{\smax}\cdot \frac{\lambda_2}{\smax} \langle \vec{e}, \vec{e}
      \rangle_{\genS} - \frac{n}{4\cdot \smax}
      \\
      (\text{Lem.~\ref{thm:Psi0properties}}) &= \frac{1}{16\Delta}\cdot
      \frac{1}{\smax}\cdot \frac{\lambda_2}{\smax} \cdot \Psi_0 -
      \frac{n}{4\cdot \smax}
    \end{align*}
  \end{proof}
  \begin{proof}[Lemma~\ref{thm:Psi0DropConditioning}]
    First, we find a bound for $\EV[\Psi_0(X^t)|X^{t-1}=x]$, i.e., the expected
    value of the potential if the previous state was $x$. There, we have
    \begin{align*}
      \EV[\Psi_0(X^t)|X^{t-1}=x] &=
      \Psi_0(x) - \EV[\Delta \Psi_0(X^t)|X^{t-1}=x] \\
      (\text{\footnotesize Lem.~\ref{thm:Psi0DropLambda2}}) \quad &\leq
      \left(1-\frac{2}{\gamma}\right) \cdot \Psi_0(x) + \frac{n}{4\cdot\smax}.
    \end{align*}
    Using this, we now use the tower property of iterated expectation to obtain
    \begin{align*}
      \EV[\Psi_0(X^t)] &=
      \EV[\EV[\Psi_0(X^t)|X^{t-1}]] \\
      &\leq \left(1-\frac{2}{\gamma}\right) \cdot \EV[\Psi_0(X^{t-1})] +
      \frac{n}{4\cdot\smax}.
    \end{align*}
  \end{proof}
  \begin{proof}[Lemma~\ref{thm:Psi0CriticalFactor}]
    We can rearrange the condition $\EV[\Psi_0(X^t)] \ge \psi_c$ as follows.
    \begin{align*}
      & & \EV[\Psi_0(X^t)] &\ge
      \frac{8\cdot n \cdot \Delta \cdot \smax}{\lambda_2} \\
      &\Leftrightarrow & \underbrace{\frac{\lambda_2}{32\Delta} \cdot
        \frac{1}{\smax^2}}_{=1/\gamma} \cdot \EV[\Psi_0(X^t)] &\ge
      \frac{n}{4\cdot \smax}.
    \end{align*}
    We insert this into the bound from Lemma \ref{thm:Psi0DropConditioning} to
    obtain
    \begin{align*}
      \EV[\Psi_0(X^{t+1})] &\leq \left(1 - \frac{2}{\gamma}\right)\cdot
      \EV[\Psi_0(X^t)] + \frac{n}{4\cdot \smax} \\
      &\leq \left(1 - \frac{2}{\gamma}\right)\cdot
      \EV[\Psi_0(X^t)] + \frac{1}{\gamma}\EV[\Psi_0(X^t)] \\
      &\leq \left(1 - \frac{1}{\gamma} \right) \cdot \EV[\Psi_0(X^t)].
    \end{align*}
  \end{proof}
  \begin{proof}[Lemma~\ref{thm:Psi0FactorDrop}]
    The proof is by induction on $t$, with $t = 0$ as the base case. There, we
    have
    \begin{equation*}
      \EV[\Psi_0(X^0)] \leq
      \left( 1 - \frac{1}{\gamma} \right)^0
      \cdot \EV[\Psi_0(X^0)].
    \end{equation*}
    Next, let the claim be true for a given $t$. Either there is a $t' < t$ such
    that $\EV[\Psi_0(X^{t'})] \leq \psi_c$ and we are done.  Otherwise, we can
    apply Lemma \ref{thm:Psi0CriticalFactor} to obtain.
    \begin{align*}
      \EV[\Psi_0(X^{t+1})] &\leq \left(1 - \frac{1}{\gamma}\right)\cdot
      \EV[\Psi_0(X^t)]\\
      (\text{\footnotesize Induction Hypothesis})\quad &\leq \left(1 -
        \frac{1}{\gamma}\right)^{t+1}\cdot \EV[\Psi_0(X^0)].
    \end{align*}
  \end{proof}
  \begin{proof}[Lemma~\ref{thm:T}]
    (1) We write down the inequality we want to prove and then show that the $T$
    given in the lemma makes it true. Since we use Lemma
    \ref{thm:Psi0FactorDrop} to connect the expectation value of $\Psi_0$ at
    time $T$ to its value at time $t = 0$, we first note that
    \begin{equation*}
      \Psi_0(X^0) \leq m^2,
      \tag{$*$}
      \label{eqn:Psi0max}
    \end{equation*}
    as the largest potential is obtained by the largest imbalance, i.e.,
    assigning all $m$ tasks to the slowest node. Then, we can deduce
    \begin{alignat*}{2}
      &    & \EV[\Psi_0(X^T)] &\leq \psi_c \\
      \text{\footnotesize (Lem.~\ref{thm:Psi0FactorDrop})} &\Leftarrow\quad & (1
      - 1/\gamma)^T \cdot \Psi_0(X^0) &\leq \frac{16 \Delta \cdot
        n \cdot \smax}{\lambda_2} \\
      (\text{\ref{eqn:Psi0max}}) &\Leftarrow\quad & T\cdot \ln (1-1/\gamma) +
      \ln \left(m^2\right) & \leq \ln \left( \frac{16\Delta \cdot n \cdot
          \smax}{\lambda_2} \right).
    \end{alignat*}
    Next we use Lemma \ref{thm:lambda2_Delta}, which states $\lambda_2 \leq
    n/(n-1) \cdot \Delta$, together with $n - 1 \ge n/2$ for $n \ge 2$ to
    further rewrite the bound.
    \begin{alignat*}{2}
      & & T\cdot \ln (1-1/\gamma) + \ln \left(m^2\right) & \leq \ln \left(
        \frac{16\Delta \cdot n \cdot \smax}{\lambda_2}
      \right)\\
      \text{\footnotesize (Lem.~\ref{thm:lambda2_Delta})} &\Leftarrow\quad &
      -\frac{1}{\gamma}\cdot T + 2 \ln(m) &\leq \ln \left(16(n-1)\right)
      + \ln(\smax)\\
      &\Leftarrow &
      -\frac{1}{\gamma}\cdot T + 2 \ln(m) &\leq \ln(8n)\\
      &\Leftarrow & -\frac{1}{\gamma}\cdot T + 2 \ln(m) &\leq \ln(n)
    \end{alignat*}
    This can be rearranged to yield the condition on $T$,
    \begin{equation*}
      T \ge 2\gamma \cdot \ln \left(\frac{m}{n}\right).
    \end{equation*}

    (2) To prove a lower bound on the probability that $\Psi_0(X^t) \leq 4\cdot
    \psi_c$, we prove an upper bound on the complementary event.  Let $Y$ denote
    the random variable $\Psi_0(X^t)$, with the $t$ from part (1) of this lemma,
    i.e. the $t$ for which $\EV[\Psi_0(X^t)] \leq \psi_c$.  Then, we have
    \begin{equation*}
      \PR[\Psi_0(X^t) > 4\cdot \psi_c] 
      \leq \PR\left[\Psi_0(X^t) > 4\cdot\EV[\Psi_0(X^t)]\right].
    \end{equation*}
    Applying Markov's inequality to this result immediately yields
    \begin{equation*}
      \PR[\Psi_0(X^t) > 4\cdot\EV[\Psi_0(X^t)]] \leq \frac{1}{4}.
    \end{equation*}
    Hence,
    \begin{equation*}
      \PR[\Psi_0(X^t) \leq 4\cdot \psi_c] \ge 1 - \frac{1}{4} = \frac{3}{4}.
    \end{equation*}
  \end{proof}
  \begin{proof}[Observation \ref{thm:LDeltaPsi0}]
    We omit the argument $x$ for brevity. Let us begin with the first
    inequality. Let $k$ be the index of the node for which $|e_k/s_k|$ is
    maximized. Then
    \begin{equation*}
      L_\Delta^2 = \frac{e_k^2}{s_k^2} \leq
      \frac{e_k^2}{s_k^2} \cdot s_k \leq
      \sum_{i \in V} \frac{e_i^2}{s_i} = \Psi_0.
    \end{equation*}
    The second inequality follows from
    \begin{equation*}
      \Psi_0 = \sum_{i\in V} \frac{e_i^2}{s_i} = \sum_{i\in V}
      \frac{e_i^2}{s_i^2}\cdot
      s_i
      \leq L_\Delta^2 \sum_{i\in V} s_i = L_\Delta^2 \cdot \SSum
    \end{equation*}
  \end{proof}
  \begin{proof}[Lemma~\ref{thm:psic_means_NE}]
    From Observation \ref{thm:LDeltaPsi0} we have $L_\Delta^2(x) \leq
    \Psi_0(x)$.  Hence, we have for all states with $\Psi_0(x) \leq
    4\cdot\psi_c$:
    \begin{equation*}
      L_\Delta \leq 
      16\cdot \sqrt{\frac{n\cdot \Delta}{\lambda_2}\cdot\smax}
      \leq
      8\cdot n^2 \cdot \smax =: a
    \end{equation*}
    where Corollary~\ref{thm:lambda2_simplebound} states that $1/\lambda_2 \leq
    n^2/4$.
    
    Next, we define $\varepsilon = 2/(1+\delta)$.  The condition for an \eapx
    Nash equilibrium is that for every edge $(i,j)\in E$ we have
    \begin{equation*}
      (1-\varepsilon)\cdot \frac{w_i}{s_i} \leq \frac{w_j + 1}{s_j}.
    \end{equation*}
    To prove that this is the case, note that the definition of $L_\Delta$
    ensures that
    \begin{equation*}
      \left|\frac{w_i}{s_i} - \frac{m}{\SSum} \right| \leq L_\Delta \Longrightarrow
      \frac{w_i}{s_i} \leq L_\Delta + \frac{m}{\SSum} \leq a + \frac{m}{\SSum}
    \end{equation*}
    and, analogously
    \begin{equation*} \frac{w_j}{s_j} \ge \frac{m}{\SSum} - a.
    \end{equation*}
    With this, we have
    \begin{alignat*}{2}
      & & (1-\varepsilon)\cdot \frac{w_i}{s_i} &\leq \frac{w_j + 1}{s_j} \nonumber\\
      \Leftarrow &\quad & \frac{\delta - 1}{\delta + 1}\cdot \left( a +
        \frac{m}{\SSum} \right) &\leq
      \left( \frac{m}{\SSum} - a \right) \nonumber\\
      \Leftarrow &\quad & \delta\cdot a - \frac{m}{\SSum} &\leq \frac{m}{\SSum}
      -
      \delta\cdot a  \nonumber\\
      \Leftrightarrow &\quad & m &\ge \delta\cdot a\cdot \SSum =
      8\cdot\delta\cdot \smax\cdot \SSum\cdot n^2.
    \end{alignat*}
    Thus, if the number of tasks is sufficiently high, we have an \eapx Nash
    equilibrium.
  \end{proof}
  \subsection{Proofs from Section \ref{sec:NE}}
  \label{sec:proofs-section-NE}
  \begin{proof}[Observation \ref{thm:observationsPsi1}]
    For brevity, we omit the argument $x$.  We begin with (1). This is simple
    algebra.
    \begin{align*}
      \Psi_1 &= \Phi_1 - \frac{m^2}{\SSum} - \frac{m \cdot n}{\SSum} +
      \frac{n}{4}\cdot \left(\frac{1}{\bar{s}_h} - \frac{1}{\bar{s}_a}\right) \\
      &= \sum_{i\in V} \frac{w_i\cdot (w_i+1)}{s_i} - \frac{m^2}{\SSum} -
      \frac{m \cdot n}{\SSum} +
      \frac{n}{4}\cdot \left(\frac{1}{\bar{s}_h} - \frac{1}{\bar{s}_a}\right) \\
      &= \sum_{i\in V} \left[\frac{\left( e_i + m / \SSum \cdot s_i
          \right)^2}{s_i} + \frac{e_i + m/\SSum \cdot s_i}{s_i}\right] -
      \frac{m^2}{\SSum} - \frac{m \cdot n}{\SSum} +
      \frac{n}{4}\cdot \left(\frac{1}{\bar{s}_h} - \frac{1}{\bar{s}_a}\right)\\
      &= \sum_{i \in V} \left[ \frac{e_i^2 + e_i}{s_i}\right] +
      \frac{n}{4}\cdot \left(\frac{1}{\bar{s}_h} - \frac{1}{\bar{s}_a}\right)\\
      &= \sum_{i \in V} \frac{\left( e_i + \frac{1}{2} \right)^2}{s_i} -
      \frac{n}{4\cdot \bar{s}_a}.
    \end{align*}

    Next, we prove (2). From the form in (1) it might appear that $\Psi_1(x)$
    can be negative. Note, however, that the task deviation vector $\vec{e}$
    satisfies $\sum_i e_i = 0$. We can use the technique of Lagrange
    multiplicators to find the minimum of $\Psi_1(x)$ under the constraint that
    the deviations sum to $0$. For an additional parameter $\lambda$, the so
    called \emph{Lagrange multiplier}, we define the Lagrange function
    \begin{equation*}
      \mathcal{L}(e_1, \dots, e_n; \lambda) = \sum_{i\in V} \frac{\left( e_i +
          \frac{1}{2}
        \right)^2}{s_i} - \lambda \cdot \sum_{i\in V} e_i.
    \end{equation*}
    The constrained minimum of $\Psi_1(x)$ is obtained for the solution of
    \begin{equation*} \frac{\partial \mathcal{L}}{\partial e_i} = 0, \qquad
      \frac{\partial \mathcal{L}}{\partial \lambda} = 0. \end{equation*}
    Carrying out the calculation shows that, indeed, minimum value of $\Psi_1$
    therefore is $0$.

    (3) is obtained by first rewriting $\Phi_1$ as
    \begin{equation*}
      \Phi_1 = \sum_{i\in V} \frac{w_i\cdot(w_i + 1)}{s_i} = 
      \Phi_0 + \sum_{i\in V} \frac{\ell_i}{s_i} =
      \Phi_0 + \sum_{i\in V} \frac{e_i}{s_i} + \frac{m\cdot n}{\SSum}.
    \end{equation*}
    Recall the definition of $\Psi_0 = \Phi_0 - m^2 / \SSum$. Hence
    \begin{align*}
      \Psi_1 &= \Phi_1 - \frac{m^2}{\SSum} - \frac{m\cdot n}{\SSum} +
      \frac{n}{4}\cdot \left(\frac{1}{\bar{s}_h} - \frac{1}{\bar{s}_a}\right) \\
      &= \Phi_0 + \sum_{i\in V} \frac{e_i}{s_i} - \frac{m^2}{\SSum} +
      \frac{n}{4}\cdot \left(\frac{1}{\bar{s}_h} - \frac{1}{\bar{s}_a}\right)\\
      &= \Psi_0 + \sum_{i \in V} \frac{e_i}{s_i} + \frac{n}{4}\cdot
      \left(\frac{1}{\bar{s}_h} - \frac{1}{\bar{s}_a}\right).
    \end{align*}

    (4) follows from the definition of $\Psi_1(x)$ by observing that apart from
    the term $\Phi_1(x)$, everything else is constant.
  \end{proof}
  \begin{proof}[Lemma \ref{thm:liljtilde_tight}]
    We have
    \begin{alignat*}{2}
      & & \ell_i - \ell_j &> \frac{1}{s_j} \\
      \Leftrightarrow &\quad &
      w_i \cdot s_j - w_j \cdot s_i &> s_i \nonumber \\
      \Leftrightarrow &\quad & w_i \cdot n_j \cdot \epsilon - w_j \cdot n_i
      \cdot \epsilon &> n_i \cdot
      \epsilon \\
      \Leftrightarrow &\quad &
      w_i \cdot n_j - w_j \cdot n_i &> n_i \\
      \intertext{Since the lefthand side and righthand side expressions are
        integers, this leads to} \\
      \Leftrightarrow &\quad & w_i \cdot n_j - w_j \cdot n_i &\ge n_i + 1 \\
      \Leftrightarrow &\quad &
      w_i \cdot s_j - w_j \cdot s_i &\ge s_i + \epsilon \\
      \Leftrightarrow &\quad & \ell_i - \ell_j &\ge \frac{1}{s_j} +
      \frac{\epsilon}{s_i\cdot s_j}
    \end{alignat*}
  \end{proof}
  \begin{proof}[Lemma \ref{thm:Psi1DropBound}]
    Apply \ref{thm:liljtilde_tight} to \ref{thm:Lemma35}, then bound the result.
    \begin{align*}
      \EV[\Delta \Psi_1] &\ge \sum_{(i,j)\in \tilde E} \frac{\ell_i -
        \ell_j}{\alpha \cdot d_{ij} \cdot \left(\frac{1}{s_i} +
          \frac{1}{s_j}\right)} \cdot \left[ \left(2 - \frac{2}{\alpha}\right)
        \cdot \left(\frac{1}{s_j} +
          \frac{\epsilon}{s_i\cdot s_j}\right) - \frac{2}{s_j}\right]\\
      &= \sum_{(i,j)\in \tilde E} \frac{\ell_i - \ell_j}{\alpha \cdot d_{ij}
        \cdot \left(\frac{1}{s_i} + \frac{1}{s_j}\right)} \cdot \left[
        \frac{2\epsilon}{s_i \cdot s_j} - \frac{2\epsilon}{\alpha\cdot s_i \cdot
          s_j} - \frac{2}{\alpha\cdot s_j}\right] \\
      &\ge \sum_{(i,j)\in \tilde E} \frac{\ell_i - \ell_j}{\alpha \cdot d_{ij}
        \cdot \left(\frac{1}{s_i} + \frac{1}{s_j}\right)} \cdot \left[
        \frac{2\epsilon}{s_i\cdot s_j} - \frac{2\epsilon^2}{4\cdot \smax \cdot
          s_i \cdot s_j} -
        \frac{2\epsilon}{4\cdot \smax \cdot s_j}\right] \\
      &\ge \sum_{(i,j)\in \tilde E} \frac{\ell_i - \ell_j}{\alpha \cdot d_{ij}
        \cdot \left(\frac{1}{s_i} + \frac{1}{s_j}\right)} \cdot \left[
        \frac{2\epsilon}{s_i\cdot s_j} - \frac{\epsilon}{2\cdot s_i\cdot s_j}
        -\frac{\epsilon}{2\cdot s_i\cdot s_j}\right] \\
      &= \sum_{(i,j)\in \tilde E} \frac{\ell_i - \ell_j}{\alpha \cdot d_{ij}
        \cdot \left(\frac{1}{s_i} + \frac{1}{s_j}\right)} \cdot
      \frac{\epsilon}{s_i\cdot s_j}\\
      &\ge \sum_{(i,j) \in \tilde E} \frac{\epsilon}{4\cdot \smax \cdot d_{ij}
        \cdot 2\cdot \smax} \cdot \frac{1}{s_j}
      \\
      &\ge \sum_{(i,j) \in \tilde E} \frac{\epsilon^2}{8 \cdot \Delta \smax^3} =
      |\tilde E| \cdot \frac{\epsilon^2}{8\Delta\cdot \smax^3} \ge
      \frac{\epsilon^2}{8\Delta \cdot \smax^3}.
    \end{align*}
  \end{proof}
  \begin{proof}[Lemma \ref{thm:Psi1Psi0Relation}]
    Observation \ref{thm:observationsPsi1} (3) states that with $\vec{e}$
    denoting the task deviation vector and $\vec{1}$ denoting the vector
    $(1,\dots, 1)^\top$, we have
    \begin{equation*}
      \Psi_1(x) = \Psi_0(x) + \langle \vec{e}, \vec{1} \rangle_{\genS} 
      + \frac{n}{4}\cdot \left( \frac{1}{\bar{s}_h} - \frac{1}{\bar{s}_a} \right).
    \end{equation*}
    It remains to bound the dot-product in the equation above.  Since $\langle
    \cdot, \cdot \rangle_{\genS}$ is an inner product, it obeys the
    Cauchy-Schwarz inequality and we have
    \begin{equation*}
      | \langle \vec{e}, \vec{1} \rangle_{\genS}|^2 \leq 
      \langle \vec{e}, \vec{e} \rangle_{\genS} \cdot \langle \vec{1}, \vec{1} \rangle_{\genS}
      = \Psi_0(x) \cdot \sum_{i\in V} \frac{1}{s_i}
      = \Psi_0(x) \cdot \frac{n}{\bar{s}_h}.
    \end{equation*}
  \end{proof}
  \begin{proof}[Lemma \ref{thm:Supermartingale}]
    The first part is obtained from simply inserting the definition.  Since for
    times $t \leq T$ the system is not in a Nash equilibrium, we can use
    Corollary \ref{thm:Psi1DropBound} to write
    \begin{equation*}
      \begin{aligned}
        \EV[Z_t | Z_{t-1} = z] &=
        \EV[\Psi_1(X^t)| \Psi_1(X^{t-1}) + (t-1) V = z] \\
        &\leq \Psi_1(X^{t-1}) - V + tV = z - (t-1)V - V + tV = z.
      \end{aligned}
    \end{equation*}
    For the second part, note that
    \begin{equation*}
      \EV[Z_t] = \EV[\EV[Z_t | Z_{t-1} = z]] \leq 
      \EV[Z_{t-1}].
    \end{equation*}
  \end{proof}
  \begin{proof}[Corollary \ref{thm:martingale}]
    For time steps $t \leq T$, we have $t \wedge T = t$ and we can just use
    Lemma \ref{thm:Supermartingale}.  Note that the expected constant drop in
    potential $\Psi_1$ only depends on the \emph{current} value of the
    potential. Hence
    \begin{equation*}
      \EV[Z_t | Z_0, Z_1, \dots Z_{t-1}] = 
      \EV[Z_{t\wedge T} | Z_{t-1}=z] \leq z.
    \end{equation*}
    For time steps $t > T$, we have
    \begin{equation*}
      \EV[Z_t | Z_0, Z_1, \dots Z_{t-1}] =
      \EV[Z_T | Z_0, \dots Z_{T-1}] =
      \EV[Z_T | Z_{T-1} = z] \leq z.
    \end{equation*}
  \end{proof}
  \begin{proof}[Corollary \ref{thm:optionalstopping}]
    First, note that the random variable $T$ is a \emph{stopping time} for
    $Z_{t\wedge T}$, because the event that $T = t$ for some time step $t$
    depends only on the state $X^t$ and, most importantly, does not depend on
    some $X^{t}$ for $t > T$.  The Optional Stopping Theorem allows us to obtain
    the claim of the corollary for a stopping time $T$ if $\EV[T] < \infty$ and
    $\EV[|Z_{t+1 \wedge T} - Z_{t\wedge T}|] < c$ for some constant $c$.  The
    first condition follows from the constant drop in the potential $\Psi_1$ as
    long as the system is not in a Nash equilibrium. The second condition
    follows from
    \begin{equation*}
      \EV[|Z_{t+1 \wedge T}-Z_{t \wedge T}|] \leq |\Delta \Psi_1(X^{t+1})| \leq 
      \max_x \Psi_1(x) + V.
    \end{equation*}
    For any given system, this expression is clearly a constant.  The Optional
    Stopping Theorem for super-martingales now states that
    \begin{equation*}
      \EV[Z_T] = \EV[Z_{T\wedge T}] \leq \EV[Z_0] = \Psi_1(X^0).
    \end{equation*}
  \end{proof}

  \section{Proofs from Section \ref{sec:weighted}}
  \begin{proof}[Lemma \ref{thm:VarianceWeighted}]
    As in the original reference, we introduce random variables $A_i$ and $C_i$
    for the tasks \emph{abandoning} and \emph{coming to} node $i$, but now they
    count the \emph{weight} of these tasks, not only their number.  For the
    $C_i$, we again split it into $Z_{ji}$ where $Z_{ji}$ counts the weight
    migrating from $j$ to $i$. Then
    \begin{equation*}
      \VAR[C_i] = \sum_{j:(j,i)\in \tilde E} \VAR[Z_{ji}].
    \end{equation*}
    \begin{equation*}
      \VAR[Z_{ji}] = \sum_{\ell\in x_j} \VAR[w_\ell^{ji}].
    \end{equation*}
    Here $w_\ell^{ji}$ is the random variable that is $w_\ell$ if task $\ell$
    moves from $j$ to $i$ and $0$ otherwise. This variable follows a
    \emph{Bernoulli distribution}. If $p$ is the probability for the event to
    occur and if $x$ is the value of the event, then the variance is
    \begin{equation*}
      \VAR[\mathsf{Ber}(x, p)] = x^2 \cdot p \cdot (1-p) \leq x^2 p.
    \end{equation*}
    This allows us to write
    \begin{align*}
      \VAR[Z_{ji}] &= \sum_{l\in x_j} \VAR[w_\ell^{ji}] \\
      &\leq \sum_{l\in x_j} w_\ell^2 \cdot \frac{\ell_j - \ell_i}{\alpha\cdot
        d_{ij}\cdot
        \left(\frac{1}{s_i} + \frac{1}{s_j}\right) \cdot W_j}\\
      &\leq \sum_{l \in x_j} w_\ell \cdot \frac{\ell_j - \ell_i}{\alpha\cdot
        d_{ij}\cdot
        \left(\frac{1}{s_i} + \frac{1}{s_j} \right) \cdot W_j}\\
      &= \frac{\ell_j - \ell_i}{\alpha \cdot d_{ij} \cdot \left(\frac{1}{s_i} +
          \frac{1}{s_j} \right)} = f_{ij},
    \end{align*}
    where we use that $w_\ell^2 \leq w_\ell$ since all tasks have weight at most
    $1$.  Hence
    \begin{equation*}
      \VAR[C_i] = \sum_{j:(j,i)\in \tilde E} f_{ij}.
    \end{equation*}
    Similarly, we define the random variable $A_{\ell i}$ that is $w_l$ if task
    $\ell$ abandons node $i$ and $0$ otherwise. It is also
    Bernoulli-distributed.
    \begin{align*}
      \VAR[A_i] &= \sum_{\ell \in x_i} \VAR[A_{\ell i}] \\
      &= \sum_{\ell \in x_i}
      \VAR\left[\mathsf{Ber}\left(w_\ell;\sum_{j:(i,j)\in\tilde E} \frac{\ell_i
            - \ell_j}{\alpha\cdot d_{ij}\cdot \left(\frac{1}{s_i} +
              \frac{1}{s_j} \right) \cdot W_i}\right)\right] \leq \sum_{\ell \in
        x_i} w_\ell^2 \cdot \sum_{j:(i,j)\in \tilde E} \frac{\ell_i -
        \ell_j}{\alpha\cdot d_{ij} \cdot \left( \frac{1}{s_i} +
          \frac{1}{s_j} \right) \cdot W_i} \\
      &\leq \sum_{j:(i,j)\in \tilde E} \sum_{\ell \in x_i} w_\ell \cdot
      \frac{\ell_i - \ell_j}{\alpha\cdot d_{ij} \cdot \left(\frac{1}{s_i} +
          \frac{1}{s_j} \right) \cdot W_i} = \sum_{j:(i,j)\in \tilde E}
      \frac{\ell_i - \ell_j}{\alpha\cdot d_{ij}\cdot \left(\frac{1}{s_i} +
          \frac{1}{s_j} \right)} = \sum_{j:(i,j)\in \tilde E} f_{ij}.
    \end{align*}
    When we add the variance of $C_i$ and $A_i$ and sum over all nodes, we get,
    in formal analogy to the unweighted case,
    \begin{equation*}
      \sum_i \frac{\VAR[W_i(X^t)|X^{t-1}=x]}{s_i} = 
      \sum_{ij} f_{ij} \left( \frac{1}{s_i} + \frac{1}{s_j} \right)
    \end{equation*}
  \end{proof}

\section{Auxiliary Results}
In this part we collect results from other papers that are essential in our
own proofs.

\begin{lemma}[Lemma~3.3 in \cite{Berenbrink2011}]
  \label{thm:Lemma35}
  For any step $t$ and any state $x$,
  \begin{equation*}
    \EV[\Delta \Phi_r(X^t) | X^{t-1} = x] \ge 
    \sum_{(i,j) \in \tilde E(x)} f_{ij}(x) \cdot 
    \left( \Lambda_{ij}^r(x) - \frac{1}{s_i} - \frac{1}{s_j} \right).
  \end{equation*}
\end{lemma}

\begin{theorem}[\cite{Weyl1912}, \cite{Bhatia2001}]
  \label{thm:Weyl}
  Let $A$, $B$ and $C = A + B$ be Hermitian matrices with eigenvalues
  $\alpha_i$, $\beta_i$ and $\gamma_i$ in ascending order. Then
  \begin{align*}
    \gamma_{i+j-1} &\ge \alpha_i + \beta_j \\
    \gamma_{i+j-n} &\leq \alpha_i + \beta_j
  \end{align*}
\end{theorem}

\begin{theorem}[Theorem B from \cite{Klyachko2000}]
  \label{thm:singular_hermition}
  The following conditions are equivalent
  \begin{enumerate}
  \item There exist matrices $A_i \in \mathrm{GL}(n, \mathbb{R})$ with given
    singular spectra
    \begin{equation*}
      \sigma_i = \sigma(A_i) \quad \text{and}\quad
      \tilde \sigma = \sigma(A_1 A_2 \cdots A_N).
      \label{eqn:Kly_cond1}
    \end{equation*}
  \item There exist symmetric $n \times n$-matrices $H_i$ with spectra
    \begin{equation*}
      \lambda(H_i) = \log \sigma_i \quad \text{and}\quad
      \lambda(H_1 + H_2 + \cdots + H_N) = \log \sigma,
      \label{eqn:Kly_cond2}
    \end{equation*}
    that is, the eigenvalues of $H_i$ are the logarithms of the singular values
    of $A_i$, and the eigenvalues of $\sum_i H_i$ are the logarithms of the
    singular values of $\prod_i A_i$.
  \end{enumerate}
\end{theorem}

\end{document}